\documentclass{amsart}
\pdfoutput=1

\input{preamble}

\def\B{\mathcal{B}}

\def\H{\mathcal{H}}

\def\rr{\mathbb{R}}

\def\romanE{\mathrm{e}}

\newcommand {\set} [1] {\ensuremath{ \left\lbrace #1 \right\rbrace }}

\newcommand{\normthree}[1]{{\left\vert\kern-0.25ex\left\vert\kern-0.25ex\left\vert #1 \right\vert\kern-0.25ex\right\vert\kern-0.25ex\right\vert}}
\newcommand {\br} [1] {\ensuremath{ \left( #1 \right) }}
\newcommand {\Br} [1] {\ensuremath{ \left[ #1 \right] }}

\newcommand {\norm} [1] {\ensuremath{ \left\| #1 \right\| }}
\newcommand {\normsub} [2] {\ensuremath{ \norm{#1}_{#2} }}

\newcommand {\abs} [1] {\ensuremath{ \left| #1 \right| }}

\newcommand {\bra} [1] {\ensuremath{ \left\langle #1 \right| }}
\newcommand {\ket} [1] {\ensuremath{ \left| #1 \right\rangle }}
\newcommand {\ketbratwo} [2] {\ensuremath{ \left| #1 \middle\rangle \middle\langle #2 \right| }}
\newcommand {\ketbra} [1] {\ketbratwo{#1}{#1}}

\newcommand {\Tr} {\ensuremath{ \mathrm{Tr} }}

\newcommand {\id} {\ensuremath{\mathds{1}}}

\newcommand{\mat}[1]{\mathbf{#1}}
\newcommand{\bsigma}{\boldsymbol{\sigma}}
\newcommand{\bomega}{\boldsymbol{\omega}}

\newcommand{\iu}{\mathrm{i}}

\newcommand{\ffmat}[4]{\begin{pmatrix}#1 & #2\\#3 & #4\end{pmatrix}}

\newcommand{\Matrix} {\ensuremath{\mathcal{M}}}
\newcommand{\PSD} {\ensuremath{\operatorname{PSD}}}
\newcommand{\Hermitian} {\ensuremath{\operatorname{H}}}
\newcommand{\Density} {\ensuremath{\operatorname{D}}}

\newcommand{\Holder} {H\"older}

\newcommand{\DepBCh} {\ensuremath{\Delta_{\bsigma,\rho}}}
\newcommand{\DepBChn} {\ensuremath{\DepBCh^{\otimes n}}}

\newcommand{\gt}[2]{\Gamma^{#1}_{\bsigma}(#2)}

\newcommand{\gtt}[2]{\bsigma^{\frac{1}{2 #1}}#2\bsigma^{\frac{1}{2 #1}}}

\newcommand{\spro}[2]{\langle #1,#2\rangle_{\bsigma}}
\newcommand{\snpro}[2]{\langle #1,#2\rangle_{\bsigma^{\otimes n}}}

\newcommand{\pn}[1]{\norm{#1}_{\bsigma, p}}
\newcommand{\qn}[1]{\norm{#1}_{\bsigma, q}}

\newcommand{\tn}[1]{\norm{#1}_{\bsigma, 2}}
\newcommand{\tnn}[1]{\norm{#1}_{\bsigma^{\otimes n}, 2}}
\newcommand{\tnnn}[1]{\norm{#1}_{\bsigma^{\otimes (n-1)}, 2}}
\newcommand{\pnn}[1]{\norm{#1}_{\bsigma^{\otimes n}, p}}
\newcommand{\pnnn}[1]{\norm{#1}_{\bsigma^{\otimes (n-1)}, p}}
\newcommand{\qnn}[1]{\norm{#1}_{\bsigma^{\otimes n}, q}}
\newcommand{\qnnn}[1]{\norm{#1}_{\bsigma^{\otimes (n-1)}, q}}

\newcommand{\pnorm}[1]{\norm{#1}_{p}}
\newcommand{\qnorm}[1]{\norm{#1}_{q}}

\newcommand{\fmat}[4]{\begin{bmatrix}#1 &\dots & #2\\
\vdots & \ddots & \vdots\\
#3 & \dots & #4\end{bmatrix}}

\newcommand{\iqp}[3]{I_{#1,#2}(#3)}

\newcommand{\ent}[2]{\mathrm{Ent}_{#1,\bsigma}(#2)}

\newcommand{\diff}[1]{\frac{\partial}{\partial #1}}

\newcommand{\diri}[2]{\mathcal{E}_{#1,\mathcal{L}}(#2)}
\newcommand{\LL}{\mathcal{L}}
\newcommand{\kk}{\mathcal{K}}
\DeclareMathOperator{\Var}{Var}

\newcommand{\nll}{\LL^{(n)}}

\newcommand{\svar}[1]{\Var_{\sigma}(#1)}

\DeclareMathOperator{\diag}{diag}

\newcommand{\PODS}{PODS}
\newcommand{\PODSname}{positive off-diagonal scaling}

\usepackage[english]{babel}
\usepackage[letterpaper,top=2cm,bottom=2cm,left=3cm,right=3cm,marginparwidth=1.75cm]{geometry}
\usepackage{graphicx}

\title{Dimension-Free Approximate Tensorization of Quantum Hypercontractivity
for Qudit Depolarizing Semigroups}

\author{Yangjing Dong$^{1}$}
\author{Li Gao$^{2,3}$}
\author{Fengning Ou$^{1}$}
\author{Penghui Yao$^{1,4}$}
\author{Haigang Zhou$^{1}$}

\thanks{$^{1}$State Key Laboratory for Novel Software Technology, Nanjing University, China.\\
$^{2}$School of Mathematics and Statistics, Wuhan University, China.\\
$^{3}$Wuhan Institute of Quantum Technology, China.\\
$^{4}$Hefei National Laboratory, Hefei 230088, China.\\
Email addresses: Yangjing Dong, \url{dongmassimo@gmail.com}; Li Gao, \url{gao.li@whu.edu.cn}; Fengning Ou, \url{reverymoon@gmail.com}; Penghui Yao, \url{phyao1985@gmail.com}; Haigang Zhou, \url{hgzhou2003@outlook.com}.}

\begin{document}

\begin{abstract}
We prove approximate tensorization for hypercontractivity and logarithmic-Sobolev constants for a class of primitive reversible quantum Markov semigroups satisfying the \PODSname\ (\PODS) condition.
This class includes qubit examples and generalized depolarizing semigroups with respect to full-rank states in arbitrary finite dimensions.
For any such semigroup $(\Phi_t)_{t\ge 0}$ and every tensor power $n$, we show that the log-Sobolev constant of the product semigroup $\Phi_t^{\otimes n}$ is at least $2/(3\ln 2)\approx 0.96$ times the log-Sobolev constant of the single-site semigroup $\Phi_t$, independently of $n$ and the local dimension $d$.
The proof first establishes an exact tensorization of the $(q,2)$-hypercontractive inequality for integer $q$, in particular $q=3$, and then extends the estimate to all real $q>2$ by complex interpolation; the standard implication from hypercontractivity to logarithmic-Sobolev inequalities yields the stated almost tensorization result.
In the qubit case, we further prove exact tensorization for primitive reversible
\PODS\ semigroups and obtain sharp $(q,2)$-hypercontractivity estimates for
generalized qubit depolarizing channels.

\end{abstract}

\maketitle

\tableofcontents

\section{Introduction}
Hypercontractivity, a notion originating from quantum field theory~\cite{nelson1966quartic,nelson1973free},
has found wide applications in mathematics such as analysis, geometry, and
probability (see e.g. the book \cite{bakry2013analysis}), computer science~\cite{21923,10.1145/174130.174138,doi:10.1137/S0097539705447372,bfd94d6a-f3c8-31f1-96b1-c865a76e1c3b},
information theory~\cite{10.1109/FOCS.2008.45,cj10-01,10.1063/1.4769269,TemmeRapidMixing2013,10.1145/3688824},
and many other areas~\cite{10.1145/1250790.1250866,bardetHypercontractivityLogarithmicSobolev2022,doi:10.1137/20M134592X,qinDecidabilityFullyQuantum2023}.
The standard hypercontractivity for Boolean functions states the following:
For
any Boolean function $f: \set{-1,1}^{n}\to\mathbb{R}$ and $1\le p\le q$
satisfying $\rho\leq\sqrt{\frac{p-1}{q-1}}$, it holds that
\begin{equation*}
    \qnorm{T_\rho^{\otimes n} f}\le \pnorm{f}.
\end{equation*}
Here, $\pnorm{\cdot}$ denotes the $L_{p}$ norm and $T_{\rho}^{\otimes n}$ is the
$n$-fold noise operator.\footnote{See \cite{o2014analysis} for a thorough
treatment.} It is well known that hypercontractivity inequalities are equivalent
to \textit{logarithmic-Sobolev inequalities} (LSI), the differential
form relating the entropy of a function to its Dirichlet forms. A logarithmic-Sobolev
inequality with parameter $\alpha$ leads to hypercontractivity with parameter
$\rho = (\frac{p-1}{q-1})^{\frac{1}{4\alpha}}$ and vice versa~\cite{gross1975logarithmic,olkiewiczHypercontractivityNoncommutativeLpSpaces1999,TemmeRapidMixing2013,kingHypercontractivitySemigroupsUnital2012,beigiQuantumReverseHypercontractivity2020,Rouze2024}.

Recently, quantum hypercontractivity has been studied extensively for its applications in quantum information and remains
an active area of research~\cite{10.1215/S0012-7094-75-04237-4,olkiewiczHypercontractivityNoncommutativeLpSpaces1999,doi:10.1142/S0219025700000030,kingHypercontractivitySemigroupsUnital2012,Temme2014Hypercontractivity,10.1063/1.4933219,beigiQuantumReverseHypercontractivity2020,beigiImprovedQuantumHypercontractivity2021,bardetHypercontractivityLogarithmicSobolev2022,baoHypercontractivityQuantumErasure2025a}.
In the classical framework, the Markov semigroup is equipped with a unique
stationary measure, and the $L_p$ norms are taken with respect to this
measure. In the quantum setting, the analogue is a quantum Markov semigroup in
the Heisenberg picture equipped with a full-rank invariant state $\bsigma$. This
state determines the natural noncommutative $L_p$ norm: the KMS (Kubo–Martin–Schwinger) inner
product is defined by (see~\cite{beigiQuantumReverseHypercontractivity2020,olkiewiczHypercontractivityNoncommutativeLpSpaces1999})
\[
    \spro{\mat{X}}{\mat{Y}}
    =
    \Tr\Br{\mat{X}^\dagger\bsigma^{1/2}\mat{Y}\bsigma^{1/2}}.
\]
The associated Kosaki weighted $L_p$-norm\cite{kosakiApplicationsComplexInterpolation1984}
is:
\[
    \norm{\mat{X}}_{\bsigma,p}
    =
    \Br{\Tr\Br{\abs{\bsigma^{1/(2p)}\mat{X}\bsigma^{1/(2p)}}^p}}^{1/p},
    \qquad 1\le p<\infty.
\]
The quantum hypercontractivity is formulated in the $\bsigma$-weighted norms rather than in the tracial Schatten norms \cite{olkiewiczHypercontractivityNoncommutativeLpSpaces1999}.

The most basic example is the quantum analogue of the classical simple
generator. For a full-rank density operator $\bsigma$, define
\[
    \LL_{\bsigma}(\mat{X})=\mat{X}-\Tr\Br{\bsigma\mat{X}}\id.
\]
The generated semigroup is
\[
    \Delta_{\bsigma,\rho}(\mat{X})
    =
    \romanE^{-t\LL_{\bsigma}}(\mat{X})
    =
    \rho\mat{X}+(1-\rho)\Tr\Br{\bsigma\mat{X}}\id,
    \qquad \rho=\romanE^{-t}.
\]
This is the generalized depolarizing semigroup in the Heisenberg-picture.
When $\bsigma=\id/d$ is the completely mixed state, it reduces
to the standard unbiased quantum depolarizing channel
\[
    \Delta^d_\rho(\mat{X})
    =
    \rho\mat{X}+(1-\rho)\Tr\Br{\mat{X}}\frac{\id}{d}.
\]

The exact logarithmic-Sobolev constants for the one-site depolarizing generator
are already completely understood in every dimension $d$~\cite[Theorem
25]{beigiQuantumReverseHypercontractivity2020},
and coincide with the classical values.
In the classical
case, hypercontractivity for a single-site generator easily extends to tensor
product semigroups, a property called \textit{tensorization}. Nevertheless, in the quantum case
tensorization is generically challenging. Before the present work, for $d=2$,
exact tensorization was known for self-adjoint unital semigroups and the generalized qubit
depolarizing channel~\cite{TemmeRapidMixing2013,kingHypercontractivitySemigroupsUnital2012,beigiQuantumReverseHypercontractivity2020,beigiImprovedQuantumHypercontractivity2021}.

For general $d$-dimensional qudit channels beyond the qubit case, far less is known.
Temme, Pastawski, and Kastoryano~\cite{Temme2014Hypercontractivity}
proved a quasi-tensorization property of the logarithmic-Sobolev constant for general primitive reversible semigroups, with constants depending on the spectral gap and the minimal eigenvalues of the invariant states. Qin
and Yao~\cite{doi:10.1137/20M134592X} proved an exact $(4, 2)$-hypercontractivity
for the tensor product of qudit depolarizing channels.

In this work, we introduce a structural property of quantum channels and quantum
Markov semigroups, which we call {\em \PODSname}\ (\PODS); see
\cref{def:pods}.
Our results concern primitive reversible QMS whose semigroup maps satisfy
\PODS.
The class of reversible \PODS\ QMS is closely related to the more familiar
class of strongly $\bsigma$-reversible QMS. Every reversible \PODS\ QMS is
strongly $\bsigma$-reversible, and the converse holds whenever $\bsigma$ has a
multiplicative Sidon spectrum. Thus the two classes coincide under this
spectral assumption; see
\cref{lem:structure-lemma-modular,lem:pods-upgrades-reversibility}.
The generalized qudit depolarizing semigroup associated with any full-rank
invariant state is a basic example.
We prove that for all $n\ge 1$, the logarithmic-Sobolev constant of the
$n$-fold tensor product $(\Phi_t^{\otimes n})_{t\ge 0}$ of a primitive reversible \PODS\ semigroup $(\Phi_t)_{t\ge 0}$ is at least
$\beta = 2/(3\ln 2)\approx 0.96$ times the single-site logarithmic-Sobolev
constant.
In contrast to previous work \cite{Temme2014Hypercontractivity}, our bound is
independent of the local dimension $d$ and the spectral
gap. The exact tensorization would have $\beta=1$, so our tensorization result is close to
optimal up to a small loss.

To the best of our knowledge, our bound gives the first
dimension-independent approximate tensorization estimate with an absolute factor.
By \cref{thm:qubit-pods-exact-tensorization}, every primitive reversible
qubit \PODS\ QMS has exact tensorization of its logarithmic-Sobolev constant.
For every full-rank qubit state $\bsigma\neq\id/2$, the spectrum is
multiplicative Sidon, so these two classes coincide. Consequently, every
primitive strongly
$\bsigma$-reversible qubit QMS satisfies
\[
  \alpha_2(\LL^{(n)})=\alpha_2(\LL).
\]
This complements the existing tensorization landscape. King
\cite{kingHypercontractivitySemigroupsUnital2012} proves exact tensorization for
self-adjoint unital qubit QMS, corresponding to the tracial invariant state \(\bsigma=\id/2\), and
Beigi--Datta--Rouzé \cite{beigiQuantumReverseHypercontractivity2020} proves exact tensorization for the generalized qubit
depolarizing semigroup. In arbitrary local dimension, the present paper gives
dimension-independent approximate tensorization for primitive reversible
\PODS\ semigroups.

Below we list a few examples as illustrations of the \PODS\ property, although not all of them satisfy hypercontractivity:
\begin{itemize}
    \item[i)] generalized depolarizing channels:
    $\Delta_{\bsigma,\rho}(\mat{X})=\rho\mat{X}+(1-\rho)\Tr(\bsigma\mat{X})\id$;
    \item[ii)] quantum Hadamard channels~\cite{bradler2010trade,watrous2018theory}: These are channels of the form $\Phi_{\mat{B}}(\mat{X})=\mat{B}\odot\mat{X}$ with $\mat{B}\succeq 0$ and $\mat{B}_{ii}=1$, which fix the diagonal entries and scale each off-diagonal matrix unit $|i\rangle\!\langle j|$ by $\mat{B}_{ij}\geq 0$;
    \item[iii)] noncommutative birth-death semigroups: Given a finite graph on the energy basis, the edge Lindbladians built from matrix units $\mat{E}_{rs}=|r\rangle\!\langle s|$ leave the diagonal algebra invariant and scale every off-diagonal matrix unit. In the path-graph case with Gibbs weights $\mu_k\propto \romanE^{-\beta k}$, this gives a finite-dimensional analogue of the bosonic Ornstein--Uhlenbeck semigroup~\cite[Section~5.3]{gao2024completepositivityorder};
    \item[iv)] as a degenerate illustration, the qubit amplitude damping channel: its diagonal subspace is invariant and the two off-diagonal matrix units are scaled by $\sqrt{1-p}$~\cite{nielsen2010quantum}. Its invariant state is pure, hence not full-rank, so it is not covered by our main full-rank reversible setting.

\end{itemize}

\subsection{Main Results}
\label{subsection:our_results}

Recall that the generalized qudit depolarizing channel is
\begin{equation*}
    \DepBCh(\mat{A}) = \rho\mat{A} + (1-\rho)\Tr\Br{\bsigma\mat{A}}\cdot\id.
\end{equation*}
The $2$-logarithmic-Sobolev constant for this semigroup
$$\alpha_{2}(\LL
_{\bsigma})= \frac{1 - 2\lambda(\bsigma)}{\ln(\lambda(\bsigma)^{-1}- 1)}$$
was computed in~\cite{beigiQuantumReverseHypercontractivity2020},  where $\LL
_{\bsigma}$ is the simple Lindblad generator for $\DepBCh$. With this
input, our first main result gives the following tensorized hypercontractivity
estimate for qudit depolarizing QMS.
\begin{theorem}\label{thm:main:depolarizing}Let $\DepBCh$ be the generalized depolarizing channel
with a full-rank
    $d\times d$ density operator $\bsigma$. Then, for
   any $1\leq p\leq q$ and every operator $\mat{A}$
    acting on $n$ qudits, we have
    \[
        \qnn{\Delta_{\bsigma,\romanE^{-t}}^{\otimes n}(\mat{A})}\leq \pnn{\mat{A}},
    \]
 where
    $$\  t= \frac{3\ln 2}{2}\cdot \frac{1}{4\alpha_{2}(\LL_{\bsigma})}
        \ln{\frac{q-1}{p-1}}.$$

\end{theorem}

\cref{thm:main:depolarizing} is a consequence of the following approximate tensorization of the LSI constant $\alpha_2$ for primitive reversible \PODS\ semigroups.
We refer to Section \ref{sec:LSI} for the detailed definition of $\alpha_2$.
\begin{theorem}\label{thm:reversible}
    Let $\LL$ be a primitive $\bsigma$-reversible Lindblad generator, and denote by
    $(\Psi_t)_{t\ge0}$ the generated quantum Markov semigroup. If $(\Psi_t)_{t\ge0}$ satisfies
    the \PODS\ property in some eigenbasis of $\bsigma$ (see \cref{sec:reversible}),
    then, for every integer $n\ge1$,
    \begin{equation}
        \alpha_2(\LL^{(n)})
        \ge \frac{2}{3\ln 2}\,\alpha_2(\LL)
        \approx 0.96\,\alpha_2(\LL).
    \end{equation}
    Here, $ \displaystyle \LL^{(n)} = \sum_{i=1}^n\id^{\otimes i-1}\otimes \LL\otimes \id^{\otimes n-i}$ is the generator of the $n$-fold tensor product semigroup $(\Psi_t^{\otimes n})_{t\ge0}$.
\end{theorem}

For qubits, the approximate factor in \cref{thm:reversible} improves to exact
tensorization.
\begin{theorem}[Exact tensorization for qubit \PODS\ QMS]
\label{thm:qubit-pods-exact-tensorization}
Let $\bsigma$ be a full-rank qubit state and let $\LL$ be a primitive
$\bsigma$-reversible Lindblad generator on $\mathcal B(\mathbb C^2)$. Write
$\Psi_t=\romanE^{-t\LL}$ and suppose that $\Psi_t$ is \PODS\ with respect to an
eigenbasis of $\bsigma$ for every $t\ge0$. Then, for every integer $n\ge1$,
\[
    \alpha_2(\LL^{(n)})=\alpha_2(\LL).
\]
\end{theorem}

To our knowledge, beyond the qubit case, the only previously known result for
 qudit channels was a quasi-tensorization of the log-Sobolev constant obtained
by~\cite{Temme2014Hypercontractivity} via the spectral gap, with constants depending on both the local dimension $d$ and the minimal eigenvalue of $\bsigma$. A brief comparison of tensorization results for QMS in the literature is presented in
\cref{tab:comparison}, where $C_{\bsigma}$ is a constant depending on the invariant state $\bsigma$.
We
will discuss these results in detail in \cref{sec:qudit}.

\begin{table}[ht]
    \centering
    \caption{Comparison of tensorization results.}
    \footnotesize
    \renewcommand{\arraystretch}{1.35}
    \begin{tabular}{ >{\centering\arraybackslash}p{3.2cm} >{\centering\arraybackslash}p{6cm}
    >{\centering\arraybackslash}p{5cm} }
        \toprule \textbf{Work}                                                         & \textbf{Setting}                        & \textbf{Tensorization Property}                                        \\
        \midrule Folklore                                                              & Classical LSI                           & Exact                                                                  \\
        \cite{cj10-01}, \cite{kingHypercontractivitySemigroupsUnital2012} & Self-adjoint unital qubit QMS           & Exact                                                                  \\
        \cite{beigiQuantumReverseHypercontractivity2020} & Generalized qubit depolarizing QMS & Exact \\
        \cite{Temme2014Hypercontractivity}                                             & Primitive reversible qudit QMS          & $\alpha_{2}(\LL^{(n)})\ge C_{\bsigma}\lambda(\LL)\geq 2C_{\bsigma}\alpha_{2}(\LL)$ \\
        \textbf{This work} & Primitive reversible qubit \PODS\ QMS & $\alpha_{2}(\LL^{(n)})=\alpha_{2}(\LL)$ (exact) \\
        \textbf{This work} & Primitive reversible qudit \PODS\ QMS & $\alpha_{2}(\LL^{(n)})\geq 0.96\alpha_{2}(\LL)$ \\
        \bottomrule
    \end{tabular}
    \label{tab:comparison}
\end{table}

We also establish in \cref{sec:q_2_qubit_hc} the optimal
$(q,2)$-hypercontractivity inequality for
the generalized qubit depolarizing channel, using the norm compression technique for
$d=2$.

\begin{theorem}[Optimal \texorpdfstring{($q,2$)}{(q,2)} HC for Generalized Qubit Depolarizing Channels]
    For $d=2$ and $2\times 2$ full-rank density operator $\bsigma$, every $q\ge2$ and every operator $\mat{X}$ acting on $n$ qubits, the generalized qubit depolarizing channel $\DepBChn$ satisfies
    \[
        \qnn{\DepBChn(\mat{X})}\leq \tnn{\mat{X}}
    \]
    whenever
    \[
        \rho\leq
        \sqrt{\frac{(1-\mu)^{\frac{2}{q}}-\mu^{\frac{2}{q}}}
        {\mu^{\frac{2}{q}-1}(1-\mu)-(1-\mu)^{\frac{2}{q}-1}\mu}}\ , \   \mu=\lambda_{\min}(\bsigma)\le 1/2.
    \]
    Moreover, this upper bound on $\rho$ is sharp.
\end{theorem}

\subsection{Application to Quantum KKL}

The Kahn--Kalai--Linial (KKL) theorem is a central result in the analysis of Boolean functions. Together with Talagrand's inequality and Friedgut's junta theorem, it forms part of a broader influence theory that has played an important role in the study of threshold phenomena, noise sensitivity, learning theory, and related questions in theoretical computer science. A quantum analogue of this circle of ideas was initiated by Montanaro and Osborne~\cite{cj10-01}, who introduced quantum Boolean functions on the qubit hypercube and studied their Fourier analysis and influences. In particular, they obtained a quantum analogue of Talagrand's inequality for the $L_2$-influences, but the corresponding quantum KKL problem remained open. Unlike in the classical Boolean case, the $L_1$- and $L_2$-influences no longer coincide, which creates a major difficulty.

Recently, Rouz{\'e}, Wirth, and Zhang~\cite{rouze2024quantumkkl} developed a systematic quantum extension of this influence theory. Their framework establishes quantum analogues of Talagrand-, KKL-, and Friedgut-type theorems for $L_1$-influences, based on suitable hypercontractivity and gradient estimates, and applies beyond the qubit hypercube to general von Neumann algebraic settings. Building on this framework, our hypercontractivity theorem provides the required hypercontractive estimate for the generalized qudit depolarizing QMS
\[
    P_t=\romanE^{-t\LL_{\bsigma}^{(n)}}=\Phi_t^{\otimes n}.
\]
Consequently, the Talagrand-, KKL-, and junta-type results of~\cite{rouze2024quantumkkl} become available for product qudit depolarizing channels of arbitrary local dimension $d$. Thus, our tensorization result yields new applications to quantum influence theory beyond the qubit setting.

\subsection{Proof Techniques}
\label{sec:proof-outline} The main theorem, \cref{thm:qudit_hc}, is proved by a bootstrapping
argument. We first prove the exact tensorization of $(q,2)$-hypercontractivity when $q$ is an integer. This
follows from a standard induction argument using a new norm compression inequality.
Then, using complex interpolation, the $(3,2)$-hypercontractivity is extended
to $(q,2)$-hypercontractivity for all real numbers $q\in[2,3]$, at the cost
of a loss in the constants. Finally, using the standard equivalence between logarithmic-Sobolev
inequalities and hypercontractivity inequalities, the $(q,2)$-hypercontractivity
is bootstrapped into full $(q,p)$-hypercontractivity for all
$1\le p\le q\le \infty$.

\noindent\textbf{Norm compression inequality.}
A crucial proof ingredient in the tensorization of the qubit case is the \textit{norm
compression inequality}. A norm compression inequality relates the norm of a
large matrix to the norms of its smaller, block-partitioned components.
King~\cite{king2003inequalities} proved the norm compression inequality for
$2\times 2$ block matrices, which shows that for a positive semi-definite
$2\times2$ block matrix and real values $p\ge2$,
\begin{equation*}
    \pnorm{\begin{pmatrix}\mat{X}&\mat{Y} \\ \mat{Y}^{*}&\mat{Z}\end{pmatrix}}\le \pnorm{\begin{pmatrix}\pnorm{\mat{X}} & \pnorm{\mat{Y}} \\ \pnorm{\mat{Y}} & \pnorm{\mat{Z}}\end{pmatrix}}.
\end{equation*}
Moreover, for $1\le p\le 2$, the above inequality holds in the reverse direction.
Audenaert~\cite{AUDENAERT2008781} generalized this result to $2\times d$ block
matrices in certain cases, including the case $p \ge 4$. However, norm
compression for $d\geq 3$ fails in general; counterexamples for real-valued $p$
have been given in~\cite{AUDENAERT2006155,AUDENAERT2008781}.

The lack of a norm compression inequality is the main obstacle to directly
extending hypercontractivity for qubit channels to general qudit channels.
Nevertheless, if we focus on the case where the parameter $p$ is an integer at
least $2$, then the norm compression inequality for general $d\times d$ block
matrices still holds. Indeed, for integer-valued exponents we can expand the
corresponding power into a sum of matrix products and apply \Holder's inequality.
As a technical tool in this work, we prove that for a general PSD $d\times d$
block matrix, and $p\ge 2$ an integer,
\begin{equation*}
    \pnorm{\fmat{\mat{A}_{11}}{\mat{A}_{1d}}{\mat{A}_{d1}}{\mat{A}_{dd}}}\le \pnorm{\fmat{\pnorm{\mat{A}_{11}}}{\pnorm{\mat{A}_{1d}}}{\pnorm{\mat{A}_{d1}}}{\pnorm{\mat{A}_{dd}}}}
    .
\end{equation*}
We further extend this norm compression inequality to the biased (weighted) norm, which
enables us to establish tensorized hypercontractivity in the biased setting.

\noindent\textbf{Tensorization of hypercontractivity.}
Given the above norm compression inequality, the proof of $(q,2)$-hypercontractivity
for tensorized depolarizing channels for integer $q$ is standard, following~\cite{kingHypercontractivitySemigroupsUnital2012}.
In particular, the base case follows from the known hypercontractivity of a single
qudit depolarizing channel~\cite{beigiQuantumReverseHypercontractivity2020}, while
the qubit tensorization result is known from~\cite{kingHypercontractivitySemigroupsUnital2012,beigiQuantumReverseHypercontractivity2020}.
This allows us to tensorize the hypercontractivity of the depolarizing channel in
the $(q,2)$ setting when $q$ is an integer. That is, we prove that for any local
dimension $d\ge2$, any $n\ge 1$, and any $d^{n}\times d^{n}$ matrix $\mat{X}$, the
following inequality holds with the parameter
$\rho = \br{q-1}^{\frac{d\ln (d-1)}{4(2-d)}}$:
\begin{equation*}
    \normsub{{(\Delta^d_{\rho})}^{\otimes n}\br{\mat{X}}}{q}\le \normsub{\mat{X}}{2}.
\end{equation*}
For $d=2$, the parameter is understood by continuity as
$\rho=(q-1)^{-1/2}$.

Nevertheless, the (exact) tensorization property achieved in \cite{king2003inequalities,
beigiQuantumReverseHypercontractivity2020} seems hard to prove in the \emph{qudit}
case, partly due to the breakdown of the norm compression property (see~\cref{remark:norm_compression_false_for_qudit}).
In~\cite{Temme2014Hypercontractivity}, Temme \textit{et al.} proved an
approximate tensorization result: they obtained a lower bound for the logarithmic-Sobolev constant
explicitly depending on the system dimension~$d$ and the minimal eigenvalues of invariant states $\bsigma$.
We overcome this difficulty by employing \emph{complex interpolation} (see~\cref{lem:com_interpolation_noncom})
on our norm compression inequality. In contrast, our result incurs only a universal constant overhead
$3\ln 2/2 \approx 1.04$, or equivalently the universal loss factor
$2/(3\ln 2)\approx 0.96$, which is independent of $d$.

\noindent\textbf{Complex interpolation.}
Since the biased norm spaces $\pn{\cdot}$ form a complex interpolation family, the estimate
\begin{equation*}
    \normsub{\Delta_{\bsigma, \rho_0}^{\otimes n}(\mat{A})}{\bsigma^{\otimes n}, 3}\le \tnn
    {\mat{A}},
\end{equation*}
immediately yields hypercontractivity for any $q\in[2,3]$,
\begin{equation}
    \label{eq:interpolated-intro}\qnn{\DepBChn(\mat{A})}\le \tnn{\mat{A}},
\end{equation}
where $\rho$ depends on $\bsigma$ and $q$. This step uses the fact that for
purely imaginary
parameter the map $\exp\br{-it\LL}$ preserves the biased $2$-norm (is a unitary), which
holds for any reversible generator.

\noindent\textbf{Log-Sobolev inequality.}
It is well known that hypercontractivity is equivalent to logarithmic-Sobolev
inequalities~\cite{olkiewiczHypercontractivityNoncommutativeLpSpaces1999}.
In our case, taking the derivative of \cref{eq:interpolated-intro} yields a logarithmic-Sobolev inequality for the simple generator $\LL_{\bsigma}$. Then, for primitive strongly reversible semigroups, aided by the Quantum
Stroock--Varopoulos inequality proved by Beigi, Datta, and Rouz{\'e}~\cite{beigiQuantumReverseHypercontractivity2020},
we can bootstrap our hypercontractivity in \cref{eq:interpolated-intro}  to a
full range hypercontractivity for all $1\le p\le q\le \infty$.
\subsection{Discussion and Open Problems}

Several problems remain open.

\noindent\textbf{Optimal constants.} We use integer tensorization of $(q,2)$-hypercontractivity and complex interpolation to overcome the lack of norm compression for general $p>2$, which was the key ingredient used in~\cite{kingHypercontractivitySemigroupsUnital2012,beigiQuantumReverseHypercontractivity2020}
for $d=2$. This comes at the cost of
introducing an absolute constant of approximately $0.96$, yielding approximate tensorization instead of exact tensorization. It is natural to ask whether the exact tensorization can hold as in the qubit case. In other words, a problem that remains open is
to determine the optimal hypercontractive parameters/Log-Sobolev constants in
local dimension $d\ge3$ for the semigroups in
\Cref{thm:main:depolarizing,thm:reversible}.

It would also be desirable to establish approximate tensorization of LSI or hypercontractivity
for more general Markov semigroups or quantum channels. One of the channels closest in spirit to the depolarizing channel is the erasure channel. Its tensorization on commutative $L_{p}$ spaces has only recently been resolved in~\cite{baoHypercontractivityQuantumErasure2025a}, and the extension of their proof to the quantum case encounters substantial difficulties. This suggests that tensorization for more general quantum channels in non-commutative spaces is challenging.

\noindent\textbf{Necessity of the PODS assumption.} Within the reversible class, our
\PODS\ assumption implies strong $\bsigma$-reversibility and is equivalent to it
whenever $\bsigma$ has a multiplicative Sidon spectrum, specifically whenever
$\bsigma\neq\id/2$ in the qubit case. Density operators with this spectral
property are dense in the state space. It is therefore natural to ask whether
approximate tensorization of LSI or hypercontractivity holds for all primitive
strongly reversible quantum Markov semigroups. A positive answer would match the
structural scope of the Quantum Stroock--Varopoulos inequality of
Beigi--Datta--Rouz\'{e}~\cite[Theorem~14]{beigiQuantumReverseHypercontractivity2020}
and place our tensorization result in the same strongly reversible framework as
their analysis of generalized qubit depolarizing semigroups.

\subsection{Organization}

In \cref{sec:preliminaries}, we introduce the  preliminaries on noncommutative $L_p$-norms, quantum Markov semigroups, hypercontractivity and log-Sobolev inequalities. In \cref{sec:qudit},
we prove the hypercontractivity of the qudit depolarizing channel and establish
the extension to \PODS\ channels, including exact tensorization in the qubit
case. \cref{sec:q_2_qubit_hc} proves the optimal
$(q,2)$-hypercontractivity result for the generalized qubit depolarizing channel. A proof
flowchart for \cref{thm:qudit_hc} is given in Figure \ref{fig:proof-roadmap}.

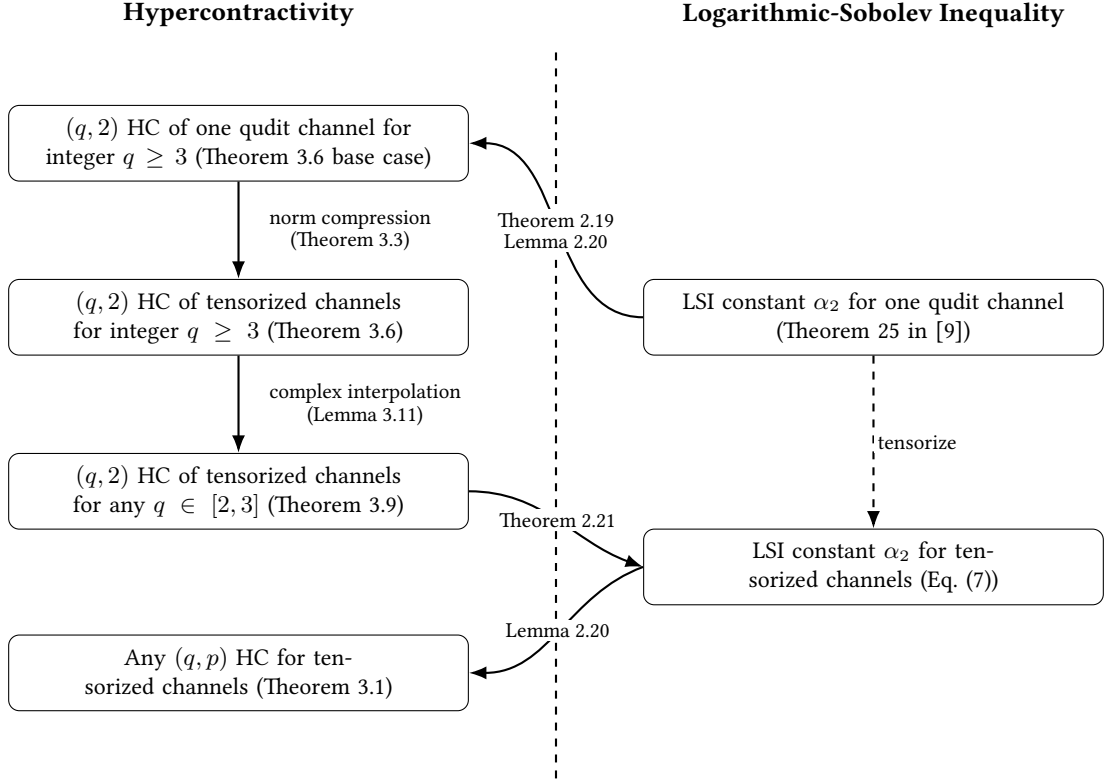
\begin{figure}[htbp]
    \centering
    \begin{tikzpicture}[
        x=1cm,
        y=1cm,
        box/.style={ draw, rounded corners, align=center, inner sep=4pt, font=\small, text width=6.3cm, minimum height=1.0cm },
        lbox/.style={ box, text width=5.8cm },
        rbox/.style={ box, text width=5.8cm },
        arr/.style={->, thick, >=Latex},
        lab/.style={font=\fontsize{8pt}{8pt}\selectfont, align=center, fill=white, inner sep=1.2pt}
    ]
        \node[font=\bfseries] at (-4.2,0.8) {Hypercontractivity};
        \node[font=\bfseries] at (4.2,0.8) {Logarithmic-Sobolev Inequality};

        \draw[dashed, thick] (0,0.3) -- (0,-9.3);

        \node[lbox]
            (B)
            at
            (-4.2,-0.9)
            {$(q,2)$ HC of one qudit channel for integer $q \geq 3$ (\cref{thm:integer_hc_noncom} base case)};

        \node[lbox]
            (C)
            at
            (-4.2,-3.2)
            { $(q,2)$ HC of tensorized channels for integer $q \geq 3$ (\cref{thm:integer_hc_noncom})};

        \node[lbox]
            (D)
            at
            (-4.2,-5.5)
            { $(q,2)$ HC of tensorized channels for any $q \in [2,3]$ (\cref{thm:complex-interpolation})};

        \node[lbox]
            (F)
            at
            (-4.2,-7.9)
            {Any $(q,p)$ HC for tensorized channels (\cref{thm:qudit_hc})};

        \node[rbox]
            (A)
            at
            (4.2,-3.2)
            {LSI constant $\alpha_{2}$ for one qudit channel\\ (Theorem~25 in \cite{beigiQuantumReverseHypercontractivity2020})};

        \node[rbox]
            (E)
            at
            (4.2,-6.5)
            {LSI constant $\alpha_{2}$ for tensorized channels (\cref{eq:lsi_tensorisation})};

        \draw[arr]
            (B) --
            node[lab, right, xshift=10pt]
                {norm compression\\(\cref{thm:qudit-norm-compression})}
            (C);

        \draw[arr]
            (C) --
            node[lab, right, xshift=10pt]
                {complex interpolation\\(\cref{lem:com_interpolation_noncom})}
            (D);

        \draw[arr, dashed] (A) -- node[lab, right] {tensorize} (E);

        \draw[arr]
            (A.west) to[out=180, in=0]
            node[lab,]
                {\cref{thm:logarithmic-Sobolev-and-hypercontractivity-equivalent}\\ \cref{lem:strook-varopoulos_non_com}}
            (B.east);

        \draw[arr]
            (D.east) to[out=0, in=160]
            node[lab,] {\cref{thm:hc_to_lsi}}
            (E.west);

        \draw[arr]
            (E.west) to[out=200, in=0]
            node[lab,] {\cref{lem:strook-varopoulos_non_com}}
            (F.east);
    \end{tikzpicture}
    \caption{Proof roadmap for \cref{thm:qudit_hc}. The proof alternates between
    hypercontractive statements and logarithmic-Sobolev inequalities.}
    \label{fig:proof-roadmap}
\end{figure}

\section{Preliminaries}\label{sec:preliminaries}
Throughout the paper, $\mathcal{H}$ is a finite-dimensional Hilbert space and $B(\mathcal{H})$ denotes the space of bounded operators on $\mathcal{H}$, which can be identified with a matrix algebra.
We use $\Matrix_{n,m}$ to denote the space of $n\times m$  complex matrices,
and use $\Matrix_{n}$ to denote the space of $n\times n$ square matrices.
We use $\Hermitian_{n}$ to denote the space of $n\times n$ Hermitian matrices, $\PSD_{n}$ to denote positive semi-definite (PSD) matrices, and $\Density_{n}$ to denote the density operators, which are PSD operators of trace one.

For $\mat{A}\in\Matrix_d$,
the trace of $\mat{A}$ is $\Tr\Br{\mat{A}} = \sum_{i=1}^d{\mat{A}_{ii}}$.
The normalized trace is $\tau\Br{\mat{A}} = \frac{1}{d}\Tr\Br{\mat{A}}$. We use $\mat{A} \preceq \mat{B}$ to denote that $\mat{B} - \mat{A}$ is positive semidefinite and $\prec$ to denote that $\mat{B} - \mat{A}$ is positive definite.
The notations $\succeq$ and $\succ$ are defined analogously. A density operator $\bsigma$ is called strictly positive if $\bsigma\succ 0$, or equivalently if it is full-rank. For a full-rank density operator $\bsigma$, we denote by $\lambda(\bsigma)$ its smallest eigenvalue.

\begin{definition}[Schatten $p$-norms] The Schatten $p$-norms and normalized Schatten $p$-norms are defined respectively as follows,
  \begin{equation*}
      \normsub{\mat{A}}{p} = \br{\Tr\Br{\abs{\mat{A}}^p}}^{1/p}\ ,\
      \normthree{\mat{A}}_p = \br{\tau\Br{\abs{\mat{A}}^p}}^{1/p}.
  \end{equation*}
\end{definition}

\subsection{Biased Norms}
Throughout this paper, $\bsigma$ denotes a full-rank density operator. The Kubo-Martin-Schwinger (KMS) inner product associated with $\bsigma$ is defined as follows (see, e.g.,~\cite{beigiQuantumReverseHypercontractivity2020,olkiewiczHypercontractivityNoncommutativeLpSpaces1999}):
\[
\spro{\mat{X}}{\mat{Y}}= \Tr\left(\mat{X}^\dagger\gtt{}{\mat{Y}} \right)\equiv \Tr\br{\mat{X}^\dagger \gt{}{\mat{Y}}}.
\]
We define the \textit{power} superoperator as $\gt{}{\mat{X}} = \gtt{}{\mat{X}}$, and we also write
\begin{equation*}
    \gt{1/p}{\mat{X}} = \gtt{p}{\mat{X}}.
\end{equation*}
We have the \textit{$\bsigma$-weighted p-norm} (or Kosaki weighted $L_p$-norm\cite{kosakiApplicationsComplexInterpolation1984}) as
\[
\pn{\mat{X}}= \Tr[|\gt{1/p}{\mat{X}}|^p]^{\frac 1p}=\Tr[|\gtt{p}{\mat{X}}|^p]^{\frac 1 p}\equiv \pnorm{\gt{1/p}{\mat{X}}}.
\]
In particular, \(\tn{\mat{X}}^2=\spro{\mat{X}}{\mat{X}}\).

 The KMS inner product and the associated norm naturally induce a family of Banach spaces known as the \textit{$\bsigma$-weighted $L_p$ spaces}. Specifically, the non-commutative space $L_p(\bsigma)$ is formally defined as the completion of the space of linear operators with respect to the $\bsigma$-weighted $p$-norm $\pn{\cdot}$. In the special case where $p=2$, the inner product $\spro{\cdot}{\cdot}$ endows $L_2(\bsigma)$ with a Hilbert space structure.

The results in this paper formulated for the $\bsigma$-norm also apply to $\bsigma^{\otimes n}$, because we may view $\bsigma^{\otimes n}$ as a full-rank density operator on the $d^n \times d^n$ matrix algebra. Consequently, when discussing the weighted $p$-norm or the inner product on operators acting on the tensor-product Hilbert space $B(\mathcal{H}^{\otimes n})$, we will, whenever no ambiguity arises, write $\spro{\cdot}{\cdot}$ for $\snpro{\cdot}{\cdot}$ and $\pn{\cdot}$ for $\norm{\cdot}_{\bsigma^{\otimes n},p}$.

The following lemmata will be useful in tensorization:

\begin{lemma}\label{lem:monotone_unbiased_norm}\label{lem:monotone_biased_norm}
    If $\mat{X}$ and $\mat{Y}$ are square matrices with non-negative entries satisfying
    $\mat{X}_{ij} \le \mat{Y}_{ij}$ for all $i, j$, then for every even integer $p \geq 2$, we have
    \[
    \normsub{\mat{X}}{p}\le \normsub{\mat{Y}}{p}.
    \]
    Moreover, if $\bsigma$ is a full-rank diagonal density operator, then
    \[
    \pn{\mat{X}}\le \pn{\mat{Y}}.
    \]
\end{lemma}

\begin{proof}
    For the unweighted case, it suffices to note that
    \[
    \normsub{\mat{X}}{p}^p
    =\sum_{i_1,\cdots,i_p=1}^n
    \mat{X}_{i_2,i_1}\mat{X}_{i_2,i_3}\mat{X}_{i_4,i_3}\mat{X}_{i_4,i_5}\cdots \mat{X}_{i_{p-2},i_{p-1}}\mat{X}_{i_p,i_1},
    \]
    which is a sum of products of non-negative entries. Hence
    \[
    \normsub{\mat{X}}{p}\le \normsub{\mat{Y}}{p}
    \]
    whenever $\mat{X}_{ij}\le \mat{Y}_{ij}$ for all $i,j$.

    For the weighted case, since $\bsigma$ is diagonal, the entries of $\gtt{p}{\mat{X}}$ and $\gtt{p}{\mat{Y}}$ are still non-negative and satisfy
    \[
    (\gtt{p}{\mat{X}})_{ij}
    =\sigma_i^{\frac{1}{2p}}\mat{X}_{ij}\sigma_j^{\frac{1}{2p}}
    \le \sigma_i^{\frac{1}{2p}}\mat{Y}_{ij}\sigma_j^{\frac{1}{2p}}
    =(\gtt{p}{\mat{Y}})_{ij}
    \qquad \text{for all } i,j.
    \]
    Applying the unweighted statement to $\gtt{p}{\mat{X}}$ and $\gtt{p}{\mat{Y}}$, we obtain
  \begin{align*}
   & \pn{\mat{X}}
    =\normsub{\gtt{p}{\mat{X}}}{p}
    \le \normsub{\gtt{p}{\mat{Y}}}{p}
    =\pn{\mat{Y}}.\qedhere
    \end{align*}
\end{proof}

If $\mat{X}$ and $\mat{Y}$ differ only on the diagonal entries, we obtain a stronger conclusion:

\begin{lemma}\label{lem:diag-monotonicity-psd}
Let \(\mat{X}\in \mathbb{C}^{d\times d}\) be positive semidefinite and let
$\mat{D}=\diag(\delta_1,\dots,\delta_d)$ be a nonnegative diagonal matrix.
Define $\mat{Y}=\mat{X}+\mat{D}$. Then, for every \(p\ge1\),
\[
\pnorm{\mat{Y}}\ge \pnorm{\mat{X}}.
\]
Moreover, if \(\bsigma\) is a full-rank diagonal density operator, then
\[
\pn{\mat{Y}}\ge \pn{\mat{X}}.
\]
\end{lemma}

\begin{proof}
Since $\mat{D}\succeq 0$, we have $\mat{Y}=\mat{X}+\mat{D}\succeq \mat{X}\succeq 0$. By Weyl's monotonicity theorem for Hermitian matrices,
\[
\lambda_i(\mat{Y})\ge \lambda_i(\mat{X})\ge 0,
\qquad i=1,\dots,d.
\]
Since $t\mapsto t^p$ is increasing on $[0,\infty)$ for every $p\ge 1$, we obtain
\[
\pnorm{\mat{Y}}^p=\sum_{i=1}^d\lambda_i(\mat{Y})^p
\ge \sum_{i=1}^d\lambda_i(\mat{X})^p
=\pnorm{\mat{X}}^p.
\]
Thus $\pnorm{\mat{Y}}\ge \pnorm{\mat{X}}$.

For the weighted statement, since $\bsigma$ is full-rank and diagonal,
\[
\gtt{p}{\mat{Y}}=\gtt{p}{\mat{X}}+\gtt{p}{\mat{D}}\succeq \gtt{p}{\mat{X}}\succeq 0,
\]
where $\gtt{p}{\mat{D}}\succeq 0$. Applying the same argument to $\gtt{p}{\mat{Y}}$ and $\gtt{p}{\mat{X}}$ gives
\[
\pn{\mat{Y}}=\pnorm{\gtt{p}{\mat{Y}}}\ge \pnorm{\gtt{p}{\mat{X}}}=\pn{\mat{X}}.
\]
\end{proof}

\begin{lemma}
\label{lem:exact_norm_compression_biased_2}
Fix full-rank density operators $\bsigma$ and $\bomega$. Suppose that
$\mat{A}=(\mat{A}_{ij})_{i,j=1}^d\in\mathbb{C}^{d\times d}\otimes\mathcal{B}(\mathcal{H})$
is written as a $d\times d$ block matrix with $\mat{A}_{ij}\in\mathcal{B}(\mathcal{H})$.
Assume that $\bsigma=\sum_{i=1}^d\sigma_i\ketbra{i}$ is diagonal in the computational basis, and define
\[
    \widetilde{\mat{A}}:=\bigl(\|\mat{A}_{ij}\|_{\bomega,2}\bigr)_{i,j=1}^d .
\]
Then $\|\widetilde{\mat{A}}\|_{\bsigma,2}=\|\mat{A}\|_{\bsigma\otimes\bomega,2}$.

In particular, if $\bomega=\bsigma^{\otimes (n-1)}$, then for every block matrix
$\mat{A}=(\mat{A}_{ij})_{i,j=1}^d$, we have
$$\big\|\bigl(\|\mat{A}_{ij}\|_{\bsigma^{\otimes(n-1)},2}\bigr)_{i,j=1}^d\big\|_{\bsigma,2} = \|\mat{A}\|_{\bsigma^{\otimes n},2}.$$
\end{lemma}

\begin{proof}
By the definition of the weighted $2$-norm, we have
$$
\|\mat{A}\|_{\bsigma\otimes\bomega,2}^2 = \|(\bsigma^{1/4}\otimes \bomega^{1/4})\mat{A}(\bsigma^{1/4}\otimes \bomega^{1/4})\|_2^2.
$$
Since $\mat{A}=(\mat{A}_{ij})_{i,j=1}^d$ and $\bsigma$ is diagonal, the $(i,j)$-block of
\[{(\bsigma^{1/4}\otimes \bomega^{1/4})\mat{A}(\bsigma^{1/4}\otimes \bomega^{1/4})} \]
is simply $\sigma_i^{1/4}\sigma_j^{1/4}\,\bomega^{1/4}\mat{A}_{ij}\bomega^{1/4}$.
Therefore, by the block decomposition of the Hilbert-Schmidt norm and the definition of the weighted $2$-norm,
$$
\|\mat{A}\|_{\bsigma\otimes\bomega,2}^2
= \sum_{i,j=1}^d \sigma_i^{1/2}\sigma_j^{1/2} \|\bomega^{1/4}\mat{A}_{ij}\bomega^{1/4}\|_2^2
= \sum_{i,j=1}^d \sigma_i^{1/2}\sigma_j^{1/2} \|\mat{A}_{ij}\|_{\bomega,2}^2.
$$

On the other hand, by the definition of $\widetilde{\mat{A}}$,
$$
\|\widetilde{\mat{A}}\|_{\bsigma,2}^2
= \|\bsigma^{1/4}\widetilde{\mat{A}} \bsigma^{1/4}\|_2^2
= \sum_{i,j=1}^d \sigma_i^{1/2}\sigma_j^{1/2} |\widetilde{\mat{A}}_{ij}|^2
= \sum_{i,j=1}^d \sigma_i^{1/2}\sigma_j^{1/2} \|\mat{A}_{ij}\|_{\bomega,2}^2.
$$

Comparing the two expressions yields $\|\widetilde{\mat{A}}\|_{\bsigma,2}^2 = \|\mat{A}\|_{\bsigma\otimes\bomega,2}^2$.
\end{proof}

Define the \textit{power operator} as:
\[
\iqp{q}{p}{\mat{X}}\coloneqq \gt{-1/q}{|\gt{1/p}{\mat{X}}|^{\frac p q}}.
\]
More properties of $\iqp{q}{p}{\mat{X}}$ can be found in~\cite{beigiQuantumReverseHypercontractivity2020}.
Using the power operator and the entropy notation defined below, we can analyze the derivative of {weighted $p$-norms} as follows:

\begin{definition}[Entropy]
    For $p \neq 0$ and $\mat{X} \succ 0$, the \textit{p-entropy} is defined\footnote{Here we follow the notations of \cite{beigiQuantumReverseHypercontractivity2020}, where the entropy differs from the usual definition in e.g.~\cite{TemmeRapidMixing2013} by an additional factor of $p$. This normalization is chosen so that the entropy agrees with the standard entropy function in the classical setting, and it also simplifies the discussion and proofs of the subsequent results.} as
    \[
    \begin{gathered}
    \ent{p}{\mat{X}}\coloneqq \Tr\left[(\gt{1/p}{\mat{X}})^p\cdot \ln{(\gt{1/p}{\mat{X}})^p} \right]
    -\Tr\left[(\gt{1/p}{\mat{X}})^p\cdot\ln \bsigma\right]-\pn{\mat{X}}^p\cdot\ln\pn{\mat{X}}^p.
    \end{gathered}
    \]
    It is straightforward to verify that $\ent{p}{I_{p,2}(\mat{X})}=\ent{q}{I_{q,2}(\mat{X})}$ for all $p,q\neq 0$.
\end{definition}

\begin{lemma}[{\cite[Derivative of the $p$-norm]{beigiSandwichedRenyiDivergence2013}}]
\label{lem:derivative_p_norm}
    For a differentiable operator-valued function $p\mapsto \mat{X}_p$ and any $p> 0$, we have
    \[
    \diff{p}\pn{\mat{X}_p} = \frac{1}{p^2}\pn{\mat{X}_p}^{1-p}\cdot\left(\frac{1}{2}\ent p{\iqp{p}{p}{\mat{X}_p}}+\frac 1 2 \ent p {I_{p,p}(\mat{X}_p^\dagger)}+\gamma\right),
    \]
    where $\gamma$ is given by
    \[
    \gamma = \frac{p^2}{2}
\left(\Tr[\gt{1/p}{\mat{Z}_p^\dagger}\cdot \gt{1/p}{\mat{X}_p}\cdot |\gt{1/p}{\mat{X}_p}|^{p-2}]+\Tr[\gt{1/p}{\mat{X}_p^\dagger}\cdot \gt{1/p}{\mat{Z}_p}\cdot |\gt{1/p}{\mat{X}_p}|^{p-2}]
\right),
    \]
    and $\mat{Z}_p = \diff{p}\mat{X}_p$.
    For $\mat{X}_p \succ 0$, the derivative simplifies to
    \[
    \diff p \pn{\mat{X}_p}=\frac{1}{p^2}\pn{\mat{X}_p}^{1-p}\cdot\left(\ent p{{\mat{X}_p}}+ p^2\Tr[\gt{1/p}{\mat{Z}_p}\cdot \gt{1/p}{\mat{X}_p}^{p-1}]\right).
    \]
    Furthermore, for fixed  $\mat{X} \succ 0$, we have
    \[
    \begin{aligned}
    \diff{p} \pn{\mat{X}}^p
    &= \pn{\mat{X}}^p\cdot \bigl(\ln \pn{\mat{X}}+p\pn{\mat{X}}^{-1}(\diff{p}\pn{\mat{X}})\bigr)\\
    &= \frac{1}{p}\bigl(\pn{\mat{X}}^p\cdot \ln \pn{\mat{X}}^p+\ent{p}{\mat{X}}\bigr)\\
    &= \frac{1}{p}\Bigl(\Tr\left[(\gt{1/p}{\mat{X}})^p\cdot \ln{(\gt{1/p}{\mat{X}})^p} \right]
        -\Tr\left[(\gt{1/p}{\mat{X}})^p\cdot\ln \bsigma\right]\Bigr).
    \end{aligned}
    \]
\end{lemma}

\subsection{Quantum Markov Semigroups}
\label{sec:reversible}

\begin{definition}
    [Quantum Markov Semigroup \& Lindbladian]
    A quantum Markov semigroup (QMS) is a continuous family $(\Psi_t)_{t\ge 0}: \B(\H)\to \B(\H)$ of completely positive unital maps of the form $\Psi_t = \romanE^{-t\LL}$, where $\LL\colon \B(\H)\to \B(\H)$ is called the \emph{Lindblad generator} (or \emph{Lindbladian}) of the QMS.
    A QMS is \emph{primitive} if its Schr\"odinger-picture adjoint has a unique
    full-rank invariant density operator.
\end{definition}

\begin{remark}
   For a fixed Lindbladian $\LL$ and fixed $\mat{X}$, the function $z\mapsto \Psi_z(\mat{X})=\romanE^{-z\LL}(\mat{X})$ is holomorphic on $\mathbb{C}$; see \cite{EngelNagel2000}.
\end{remark}

\begin{prop}[{\cite[Proposition~7]{beigiQuantumReverseHypercontractivity2020}}]
\label{prop:QMS}
Every primitive QMS $(\Psi_t)_{t\ge0}$ with full-rank invariant state $\bsigma$ is
$p$-contractive for each $p\in[1,\infty)$; equivalently,
\[
    \pn{\Psi_t(\mat{X})}\leq \pn{\mat{X}},
    \qquad \mat{X}\succ0,\quad t\ge0.
\]
\end{prop}

\begin{definition}[Reversibility]\label{def:reversibility}
  Let $(\Psi_t)_{t\ge 0}$ be a quantum Markov semigroup with generator $\LL$.
  \begin{enumerate}
      \item The semigroup $(\Psi_t)_{t\ge 0}$ is \emph{$\bsigma$-reversible} (also called KMS $\bsigma$-symmetric) if $\LL$ and each $\Psi_t$ are self-adjoint with respect to the KMS inner product $\spro \cdot \cdot$;
      \item The semigroup $(\Psi_t)_{t\ge 0}$ is \emph{strongly $\bsigma$-reversible} (also called GNS $\bsigma$-symmetric) if $\LL$ and each $\Psi_t$ are self-adjoint with respect to the GNS inner product $\langle \mat{X},\mat{Y}\rangle_{1,\bsigma}=\Tr[\bsigma \mat{X}^\dagger \mat{Y}]$.
  \end{enumerate}
\end{definition}
It is known that strongly $\bsigma$-reversible implies $\bsigma$-reversible, see e.g. \cite{carlen2017gradient}. The Lindblad generator of major interest in this work is the \textit{simple generator},
which is defined as $\LL_\bsigma(\mat{X})=\mat{X}-\Tr(\bsigma \mat{X})\id$.
The semigroup generated by the simple generator is the generalized depolarizing semigroup,
\begin{equation*}
     \romanE^{-t\LL_\bsigma}(\mat{A}) = \romanE^{-t}\mat{A} + (1-\romanE^{-t})\Tr\Br{\bsigma \mat{A}}\id,
\end{equation*}
whose Schr\"odinger-picture adjoints are the corresponding depolarizing channels.

\begin{definition}[Positive Off-Diagonal Scaling]\label{def:pods}
 Fix an orthonormal basis
  $\{\ket{i}\}_{i=1}^d$.
  We say that a linear map $\Psi:\Matrix_d\to\Matrix_d$ is \emph{\PODSname} (\PODS) with respect to this basis if there exist numbers $\gamma_{i,j} \ge 0$ for $i \neq j$ and $\eta_{i,j} \ge 0$ for $i, j$ such that
\begin{itemize}
    \item for all $i \neq j$,
    \(
    \Psi(\ketbratwo{i}{j}) = \gamma_{i,j}\ketbratwo{i}{j};
    \)
    \item for all $i$,
    \(
    \Psi(\ketbra{i}) = \sum_{j=1}^d \eta_{i,j}\ketbra{j}.
    \)
\end{itemize}
If $\Psi=\Psi_t$ is an element of a semigroup, we write the corresponding coefficients as $\gamma_{i,j}(t)$ and $\eta_{i,j}(t)$.
\end{definition}

For instance, if $\bsigma=\sum_{k=1}^d\sigma_k\ketbra{k}$ is diagonal in the
chosen basis, then the generalized depolarizing semigroup
$\Phi_t(\mat{A})=\romanE^{-t}\mat{A}+(1-\romanE^{-t})\Tr\Br{\bsigma \mat{A}}\id$ is \PODS\ with
\[
    \gamma_{i,j}(t)=\romanE^{-t},\qquad
    \eta_{i,j}(t)=\romanE^{-t}\delta_{ij}+(1-\romanE^{-t})\sigma_i.
\]
Further examples are listed in the Introduction, and
\cref{lem:structure-lemma-modular} below shows that the \PODS\ property also
follows from strong reversibility under a multiplicative-Sidon condition on the
reference state.

We can now define the \textit{tensorization} of generators.
Given two Lindblad generators $\LL$ and $\kk$ associated with the semigroups $\set{\Phi_t\colon t \ge 0}$ and $\set{\Psi_t\colon t \ge 0}$, respectively,
their tensor product is defined as $\LL \otimes \id + \id \otimes \kk$, which generates the semigroup $\set{\Phi_t \otimes \Psi_t\colon t \ge 0}$.
Furthermore, if we let
\[
\LL_i = \id^{\otimes i-1}\otimes \LL\otimes \id^{\otimes n-i}, \LL^{(n)} = \sum_{i=1}^n\LL_i,
\]
then we have $$\Phi_t^{\otimes n}=\romanE^{-t\sum_{i=1}^n\LL_i}=\romanE^{-t\LL^{(n)}}.$$
In this work, when $\bsigma$ is clear from context, we use $\LL_\bsigma$ to denote the simple generator,
while $\LL$ denotes a general Lindblad generator. We use $\Phi_t = \Delta_{\bsigma, \romanE^{-t}}$ to denote the generalized depolarizing semigroup, while $\Psi_t$ denotes a general semigroup generated by $\LL$. We will occasionally use the relation $\rho = \romanE^{-t}$ to pass between the equivalent time and noise parameters appearing in different conventions for $\DepBCh$ and $\Phi_t$.

It is known that $(\Phi_t)_{t\ge0}$ is primitive and that $\LL_\bsigma$ and
$\Phi_t$ are strongly $\bsigma$-reversible, hence also $\bsigma$-reversible.
The same reversibility properties hold for $\nll_\bsigma$, $\LL_i$, and
$\Phi_t^{\otimes n}$, and the product QMS $(\Phi_t^{\otimes n})_{t\ge0}$ is
also primitive; see \cite{beigiQuantumReverseHypercontractivity2020}. The
following standard facts will be useful in our proof:

\begin{lemma}\label{lem:complex-depolarize-2norm-invariant_noncom}
Suppose that $\LL$ is the $\bsigma$-reversible generator of $(\Psi_t)_{t\ge0}$, and denote by
$\LL^{(n)}$ the generator of $\Psi_t^{\otimes n}$. Then, for every $s\in\mathbb R$ and every operator $\mat{X}$,
\[
\tnn{\Psi_{\iu s}^{\otimes n}(\mat{X})} \equiv \tnn{\romanE^{-\iu s\LL^{(n)}}(\mat{X})} = \tnn{\mat{X}}.
\]
\end{lemma}

\begin{proof}
This is a standard fact following from the fact that $\LL^{(n)}$ is self-adjoint with respect to the $\bsigma^{\otimes n}$-inner product.
Indeed, each $\LL_i$ is self-adjoint with respect to the $\bsigma^{\otimes n}$-inner product. Hence $\LL^{(n)}=\sum_{i=1}^n\LL_i$ inherits this property, and we have
\begin{align*}
\tnn{\Psi_{\iu s}^{\otimes n}(\mat{X})}^2
&=\snpro{\romanE^{-\iu s\LL^{(n)}}(\mat{X})}{\romanE^{-\iu s\LL^{(n)}}(\mat{X})}\\
&=\snpro{\mat{X}}{(\romanE^{-\iu s\LL^{(n)}})^\dagger \circ \romanE^{-\iu s\LL^{(n)}}(\mat{X})}\\
&= \snpro{\mat{X}}{\romanE^{\iu s\LL^{(n)}}\circ \romanE^{-\iu s\LL^{(n)}}(\mat{X})}\\
&= \snpro{\mat{X}}{\mat{X}}\\
&=\tnn{\mat{X}}^2.    \qedhere
\end{align*}
\end{proof}

\subsection{Sidon Sets and Multiplicative Sidon Spectrum}

\begin{definition}[Sidon Sets]
  Let $G$ be a group with group action $+$,
  and let $S$ be a subset of $G$.
  $S$ is a Sidon set (or Golomb ruler) if
  \begin{equation*}
      \forall x_i, x_j, x_k, x_\ell\in S: x_i + x_j = x_k + x_\ell \implies \set{x_i, x_j} = \set{x_k, x_\ell}.
  \end{equation*}
    Informally, this means that the sums of distinct pairs of elements in $S$ are all distinct.
\end{definition}

In this work, we consider $\bsigma$ with eigenvalues that form Sidon sets under multiplication.
\begin{definition}[Multiplicative Sidon Spectrum]\label{def:multiplicative-sidon-spectrum}
  Let $\bsigma$ be a full-rank density operator.
  We say that $\bsigma$ has a multiplicative Sidon spectrum if the eigenvalues of $\bsigma$
  are all distinct and form a Sidon set under multiplication.
\end{definition}

\begin{lemma}\label{lem:structure-lemma-modular}
Assume that $\bsigma$ is a full-rank density operator with a multiplicative Sidon spectrum.
If $(\Psi_t)_{t\ge0}$ is a strongly $\bsigma$-reversible quantum Markov semigroup, then
$\Psi_t$ is \PODS\ with respect to the eigenbasis of $\bsigma$ for every $t\ge0$.
\end{lemma}

\begin{proof}[Proof of \cref{lem:structure-lemma-modular}]
Without loss of generality, we assume $\bsigma$ is a diagonal density operator.
Fix $t \ge 0$, and write $\mat{E}_{ij} := \ketbratwo{i}{j}$.
Since $(\Psi_t)_{t \ge 0}$ is strongly $\bsigma$-reversible, by the standard characterization of GNS symmetry (\cite[Lemma 12]{beigiQuantumReverseHypercontractivity2020}), $\Psi_t$ is KMS-symmetric and commutes with the modular automorphism group:
\[
\Psi_t \circ \Delta_\bsigma^{\iu s} = \Delta_\bsigma^{\iu s} \circ \Psi_t
\qquad
\text{for all } s \in \mathbb{R},
\]
where $\Delta_\bsigma(\mat{X}) = \bsigma \mat{X} \bsigma^{-1}$.

Because $\bsigma$ is diagonal, each matrix unit $\mat{E}_{ij}$ is an eigenvector of $\Delta_\bsigma$:
\[
\Delta_\bsigma^{\iu s}(\mat{E}_{ij}) = \Bigl(\frac{\lambda_i}{\lambda_j}\Bigr)^{\iu s} \mat{E}_{ij}.
\]
Hence, for $i \neq j$, the vector $\mat{E}_{ij}$ belongs to the eigenspace corresponding to the character $s \mapsto (\lambda_i/\lambda_j)^{\iu s}$.
The multiplicative Sidon condition implies that the off-diagonal ratio spectrum is simple. Indeed, if
\[
    \lambda_i/\lambda_j=\lambda_k/\lambda_\ell,\qquad i\ne j,\ k\ne \ell,
\]
then \(\lambda_i\lambda_\ell=\lambda_k\lambda_j\), and by the Sidon property
\(\{i,\ell\}=\{k,j\}\). Since \(i\ne j\) and \(k\ne \ell\), this forces
\((i,j)=(k,\ell)\).
Thus this eigenspace is one-dimensional.

Since $\Psi_t$ commutes with $\Delta_\bsigma^{\iu s}$ for every $s$, it preserves each such eigenspace.
Therefore, for every $i \neq j$, there exists a scalar $\gamma_{i,j}(t)$ such that
\[
\Psi_t(\mat{E}_{ij}) = \gamma_{i,j}(t) \mat{E}_{ij}.
\]

For the diagonal part, note that $\Delta_\bsigma^{\iu s}(\mat{E}_{ii}) = \mat{E}_{ii}$ for every $i$ and $s$.
Thus the fixed-point space of $\Delta_\bsigma^{\iu s}$ is exactly the diagonal subalgebra
\[
\mathrm{Diag} := \mathrm{span}\set{\mat{E}_{11}, \dots, \mat{E}_{dd}}.
\]
Since $\Psi_t$ commutes with $\Delta_\bsigma^{\iu s}$, it preserves this subspace.
Hence for each $i$ there exist coefficients $\eta_{i,j}(t)$ such that
\[
\Psi_t(\mat{E}_{ii}) = \sum_{j=1}^d \eta_{i,j}(t) \mat{E}_{jj}.
\]

We next show that $\gamma_{i,j}(t)$ is real.
Since $\Psi_t$ is GNS self-adjoint with respect to $\langle \mat{X},\mat{Y}\rangle_{1,\bsigma} = \Tr[\bsigma \mat{X}^\dagger \mat{Y}]$, we have
\[
\langle \Psi_t(\mat{E}_{ij}), \mat{E}_{ij}\rangle_{1,\bsigma}
=
\langle \mat{E}_{ij}, \Psi_t(\mat{E}_{ij})\rangle_{1,\bsigma}.
\]
Substituting $\Psi_t(\mat{E}_{ij}) = \gamma_{i,j}(t) \mat{E}_{ij}$ gives
\[
\overline{\gamma_{i,j}(t)} \, \langle \mat{E}_{ij}, \mat{E}_{ij}\rangle_{1,\bsigma}
=
\gamma_{i,j}(t) \, \langle \mat{E}_{ij}, \mat{E}_{ij}\rangle_{1,\bsigma}.
\]
Since $\langle \mat{E}_{ij}, \mat{E}_{ij}\rangle_{1,\bsigma} = \lambda_j > 0$, it follows that $\gamma_{i,j}(t) \in \mathbb{R}$.

We now prove $\gamma_{i,j}(t) \ge 0$.
Applying the above argument to $\Psi_{t/2}$, we obtain a real number $\gamma_{i,j}(t/2)$ such that
\[
\Psi_{t/2}(\mat{E}_{ij}) = \gamma_{i,j}(t/2) \mat{E}_{ij}.
\]
Using the semigroup property,
\[
\Psi_t(\mat{E}_{ij})
=
\Psi_{t/2}\bigl(\Psi_{t/2}(\mat{E}_{ij})\bigr)
=
\gamma_{i,j}(t/2)^2 \mat{E}_{ij}.
\]
Therefore $\gamma_{i,j}(t) = \gamma_{i,j}(t/2)^2 \ge 0$.

Finally, since $\Psi_t$ is a positive map and each $\mat{E}_{ii}$ is positive semidefinite, the operator $\Psi_t(\mat{E}_{ii})$ is also positive semidefinite.
As $\Psi_t(\mat{E}_{ii})$ is diagonal, all its diagonal entries must be nonnegative.
Hence $\eta_{i,j}(t) \ge 0$ for all $i,j$.
This completes the proof.
\end{proof}

\begin{lemma}\label{lem:pods-upgrades-reversibility}
Let $\bsigma=\sum_{i=1}^d\sigma_i\ketbra{i}\succ0$, and let
$\Psi:\Matrix_d\to\Matrix_d$ be self-adjoint with respect to the KMS
$\bsigma$-inner product. If $\Psi$ is \PODS\ with respect to the displayed
eigenbasis of $\bsigma$, then $\Psi$ is self-adjoint with respect to the GNS
inner product $\langle\cdot,\cdot\rangle_{1,\bsigma}$. Consequently, every
$\bsigma$-reversible QMS $(\Psi_t)_{t\ge0}$ such that $\Psi_t$ is \PODS\ with
respect to an eigenbasis of $\bsigma$ for every $t\ge0$ is strongly
$\bsigma$-reversible.
\end{lemma}

\begin{proof}
Write $\mat{E}_{ij}=\ketbratwo{i}{j}$. In the chosen eigenbasis,
\[
  \langle \mat{E}_{ij},\mat{E}_{k\ell}\rangle_{\bsigma}
  =\sqrt{\sigma_i\sigma_j}\,\delta_{ik}\delta_{j\ell},
  \qquad
  \langle \mat{E}_{ij},\mat{E}_{k\ell}\rangle_{1,\bsigma}
  =\sigma_j\,\delta_{ik}\delta_{j\ell}.
\]
Thus the diagonal subspace and the one-dimensional off-diagonal subspaces
$\mathbb{C}\mat{E}_{ij}$, $i\ne j$, are mutually orthogonal for both inner
products. On the diagonal subspace the two inner products coincide. By the
\PODS\ property, $\Psi$ preserves the diagonal subspace and acts on every
$\mathbb{C}\mat{E}_{ij}$, $i\ne j$, as multiplication by a real number
$\gamma_{i,j}\ge0$. Hence KMS self-adjointness on the diagonal subspace is
exactly GNS self-adjointness there, while the action on each off-diagonal
one-dimensional subspace is GNS self-adjoint. Therefore $\Psi$ is GNS
self-adjoint.

Applying this argument to every $\Psi_t$ shows that the entire semigroup is
GNS self-adjoint. Finally,
$\LL=-\left.\frac{\mathrm d}{\mathrm dt}\Psi_t\right|_{t=0}$ is GNS
self-adjoint as well.
\end{proof}

\begin{remark}[Comparison of the two directions]
The preceding two lemmata give implications in opposite directions:
\[
\begin{aligned}
\substack{\text{strongly $\bsigma$-reversible QMS}\\
          {}+\text{ multiplicative Sidon spectrum of $\bsigma$}}
&\Longrightarrow \text{\PODS\ in the eigenbasis of $\bsigma$},\\
\substack{\text{$\bsigma$-reversible QMS}\\
          {}+\text{ \PODS\ in an eigenbasis of $\bsigma$}}
&\Longrightarrow \text{strongly $\bsigma$-reversible QMS}.
\end{aligned}
\]
The spectral assumption appears only in the first implication.  Thus, when
$\bsigma$ has a multiplicative Sidon spectrum, a QMS is strongly
$\bsigma$-reversible if and only if it is both $\bsigma$-reversible and
\PODS\ with respect to the eigenbasis of $\bsigma$.
\end{remark}

\begin{remark}
By \cref{lem:structure-lemma-modular}, whenever $\bsigma$ has a multiplicative
Sidon spectrum, every strongly $\bsigma$-reversible QMS is \PODS.
In particular, when $d=2$, any full-rank state $\bsigma\neq \id/2$ automatically has a multiplicative Sidon spectrum.

Let us spell out a concrete noncommutative birth-death example covered by
\cref{lem:structure-lemma-modular}, following~\cite[Section~5.3]{gao2024completepositivityorder}.
Take the path graph on $V=\{1,\ldots,d\}$ and fix $\beta>0$ and
\[
    \bsigma=\sum_{k=1}^d \mu_k \mat{E}_{kk},
    \qquad
    \mu_k=Z^{-1}\romanE^{-\beta 2^k}.
\]
For each edge $(r,s)$, set $\romanE^{\beta_{rs}}=\mu_s/\mu_r$ and define
\[
    \LL_{rs}(\mat{X})
    =
    \romanE^{\beta_{rs}/2}(\mat{E}_{ss}\mat{X}+\mat{X}\mat{E}_{ss}-2\mat{E}_{sr}\mat{X}\mat{E}_{rs})
    +
    \romanE^{-\beta_{rs}/2}(\mat{E}_{rr}\mat{X}+\mat{X}\mat{E}_{rr}-2\mat{E}_{rs}\mat{X}\mat{E}_{sr}).
\]
Then $\LL_{\mathrm{BD}}=\sum_{k=1}^{d-1}\LL_{k,k+1}$ is strongly
$\bsigma$-reversible. Moreover, the diagonal algebra is invariant and evolves by
a classical birth-death chain, while for $i\neq j$ each matrix unit $\mat{E}_{ij}$ is an
eigenvector of $\LL_{\mathrm{BD}}$. Hence the semigroup generated by
$\LL_{\mathrm{BD}}$ is \PODS. Finally, the spectrum of $\bsigma$ is multiplicative
Sidon: if $\mu_i\mu_j=\mu_k\mu_\ell$, then
$2^i+2^j=2^k+2^\ell$, and uniqueness of binary expansion gives
$\{i,j\}=\{k,\ell\}$.

We emphasize, however, that the multiplicative Sidon spectrum is clearly not necessary for \PODS. Indeed, the generalized depolarizing channel is \PODS\ for any diagonal $\bsigma$, without requiring its spectrum to be multiplicative Sidon.
Finally, in the sequel, the \PODS\ property will be the key structural property used to prove the tensorization of integer hypercontractivity in \cref{thm:integer_hc_noncom}.

\end{remark}

\subsection{Logarithmic-Sobolev Inequality}\label{sec:LSI}
Let $p\in [1,\infty]$ and $\frac{1}{p}+\frac{1}{\hat p}=1$.
For a primitive $\bsigma$-reversible Lindbladian $\LL$, the \textit{$p$-Dirichlet form} is defined as\footnote{Similarly, the definition of the Dirichlet form here differs from that in~\cite{TemmeRapidMixing2013} by an additional factor of $-p/2$.}
\[
\diri{p}{\mat{X}}=\frac{p\hat p}{4}\spro{\iqp{\hat{p}}{p}{\mat{X}}}{\mathcal{L}(\mat{X})}.
\]
We say that $\LL$ satisfies the \textit{$p$-logarithmic-Sobolev} inequality if there is an $\alpha>0$ such that
\[
\alpha\ \ent p {\mat{X}}\leq \diri{p}{\mat{X}},\forall \mat{X} \succ 0.
\]
The best constant $\alpha$ is called the \textit{$p$-logarithmic-Sobolev inequality} ($p$-LSI) constant, denoted \footnote{The definition of the constant used here also differs from the conventional one by a factor of $1/2$. We have considered this difference when comparing our results with those in~\cite{Temme2014Hypercontractivity}.} by $\alpha_p(\LL)$:
\begin{align}\label{def:LSI}
\alpha_p(\LL)\coloneqq \inf_{\mat{X}\succ 0, \ent p {\mat{X}}\neq 0}{\frac{\diri{p}{\mat{X}}}{\ent p {\mat{X}}}}.
\end{align}

There is an equivalence between the logarithmic-Sobolev inequality and hypercontractivity. To be concrete, if the $p$-logarithmic-Sobolev constant admits a uniform lower bound for all $p>1$, one can derive hypercontractivity as follows:

\begin{theorem}[\cite{beigiQuantumReverseHypercontractivity2020,olkiewiczHypercontractivityNoncommutativeLpSpaces1999}]\label{thm:logarithmic-Sobolev-and-hypercontractivity-equivalent}
For a primitive $\bsigma$-reversible Lindblad generator $\LL$, assume that
$\alpha:=\inf_{r\ge1}\alpha_r(\LL)>0$. If $1\le p\le q$ and
\[
    t\geq \frac{1}{4\alpha}\ln{\frac{q-1}{p-1}},
\]
then the generated QMS $(\Psi_t)_{t\ge0}$ satisfies
\[
    \qn{\Psi_t(\mat{X})}\leq \pn{\mat{X}},\qquad \mat{X}\succeq0.
\]
\end{theorem}

In general, it is difficult to determine the logarithmic-Sobolev constant, or even to establish the positivity of~$\alpha$. Fortunately, Beigi et al.~\cite{beigiQuantumReverseHypercontractivity2020} showed that the logarithmic-Sobolev constant of the simple generator $\LL_\bsigma$ is monotone decreasing with respect to $p\in [1,2]$, which implies that the constant $\alpha$ for the simple generator $\LL_\bsigma$ equals $\alpha_2(\LL_\bsigma)$:

\begin{lemma}[{\cite[Theorem 14]{beigiQuantumReverseHypercontractivity2020}}]\label{lem:strook-varopoulos_non_com}
    For a strongly $\bsigma$-reversible Lindblad generator $\LL$ and every $\mat{X}\succ0$,
    \[
    \mathcal{E}_{\hat p}\bigl(I_{\hat p,2}(\mat{X})\bigr)
    = \mathcal{E}_{p}\bigl(I_{p,2}(\mat{X})\bigr)
    \ge \mathcal{E}_{q}\bigl(I_{q,2}(\mat{X})\bigr),
    \qquad \forall\, 0 < p \le q \le 2,
    \]
    where $\hat{p}^{-1} + p^{-1} = 1$. If, in addition, $\LL$ is primitive,
    then $\alpha_2(\LL) = \inf_{p\geq 1}\alpha_p(\LL)$.
\end{lemma}

Moreover, the value of $\alpha_2(\LL_\bsigma)$ is given by
    \[
        \alpha_2(\mathcal{L}_\bsigma) = \frac{1 - 2\lambda(\bsigma)}{\ln(\lambda(\bsigma)^{-1} - 1)},
    \]
    where $\lambda(\bsigma)$ denotes the smallest positive eigenvalue of $\bsigma$ (see Theorem 25 in \cite{beigiQuantumReverseHypercontractivity2020}).\footnote{When $\lambda(\bsigma)\to 1/2$, applying l’Hôpital’s rule yields $\alpha_2(\mathcal{L}) = 1/2$, which is consistent with the result of~\cite{kingHypercontractivitySemigroupsUnital2012}.}

Conversely, if we have hypercontractivity, we can also derive a lower bound on the LSI constant. This implication has been discussed extensively in the literature \cite{beigiHypercontractivityLogarithmicSobolev2016,kingHypercontractivitySemigroupsUnital2012,olkiewiczHypercontractivityNoncommutativeLpSpaces1999}
; for completeness, we provide a proof here.

\begin{theorem}\label{thm:hc_to_lsi}
    Suppose that $\LL$ is a primitive $\bsigma$-reversible Lindblad generator generating the QMS $(\Psi_t)_{t\ge0}$. Fix $p\ge1$, and assume that $t_0=t_0(q)\ge0$ is continuously differentiable near $q=p$, satisfies $t_0(p)=0$, and obeys
    \begin{equation}
       \qn{\Psi_{t_0}(\mat{A})}\leq \pn{\mat{A}},\forall \mat{A} \succ 0, q\geq p,\label{eq:hc_ad_hoc}
    \end{equation}
    Then, for every $\mat{X}\succ0$,
    \[
\ent{p}{\mat{X}}\le 4(p-1)\left.\frac{\mathrm{d}t_0}{\mathrm{d}q}\right|_{q=p}\,\mathcal{E}_{p,\LL}(\mat{X}),
    \]
thus
\[
\alpha_p(\LL)\ge \left.\frac{\mathrm{d}q}{\mathrm{d}t_0}\right|_{q=p}\cdot\frac{1}{4(p-1)}\quad(p>1).
\]
\end{theorem}

\begin{proof}
Regard the left-hand side of \cref{eq:hc_ad_hoc} as a function of $q$.
Since equality holds at $q=p$, its right derivative at $q=p$ is non-positive.\footnote{Here we assume differentiability at $q=p$; a rigorous justification can be found in~\cite{beigiHypercontractivityLogarithmicSobolev2016}.} Denoting $\mat{B} = \Psi_{t_0}\br{\mat{A}}=\romanE^{-t_0\LL}(\mat{A})$, by \cref{lem:derivative_p_norm},
    \[
    \diff{q}\qn{\mat{B}} = \frac{1}{q^2}\qn{\mat{B}}^{1-q}(\ent{q}{\mat{B}}+q^2\Tr[\gt{1/q}{\mat{Z}}\cdot \gt{1/q}{\mat{B}}^{q-1}]),
    \]
    where
    $$\mat{Z} = \diff q \mat{B}=-\frac{\mathrm{d}t_0}{\mathrm{d}q}\cdot\LL\br{\mat{B}}.$$
    It has the same sign as
    \[
    \begin{aligned}
    &\ent{q}{\mat{B}}-q^2\frac{\mathrm{d}t_0}{\mathrm{d}q}\Tr[\gt{1/q}{\LL\br{\mat{B}}}\cdot \gt{1/q}{\mat{B}}^{q-1}]\\
    =&\ent{q}{\mat{B}}  -4(q-1)\frac{\mathrm{d}t_0}{\mathrm{d}q}\cdot \mathcal{E}_{q,\LL}\br{\mat{B}}.
    \end{aligned}
    \]
    Setting $q=p$, and using $t_0(p)=0$, gives the desired estimate.
\end{proof}

\section{Qudit Hypercontractivity}
\label{sec:qudit}
In this section, we derive the hypercontractive inequality for qudit depolarizing channels, as outlined
in \cref{sec:proof-outline}.
The main result of this section is a nearly optimal approximate tensorization property
of the logarithmic-Sobolev constant. Indeed, for the constant $\beta = \frac{2}{3\ln
2}\approx 0.96$, we show that the logarithmic-Sobolev constant for the $n$-fold tensor
product of any primitive reversible \PODS\ QMS is at least
\begin{equation*}
  \alpha_2(\LL^{(n)}) \ge \beta\cdot\alpha_2(\LL).
\end{equation*}
We first state this result in the form of a hypercontractivity inequality.
\begin{theorem}[Main Theorem]
  \label{thm:qudit_hc}
  Fix $1\leq p\leq q$, an integer $n\ge1$, and a full-rank $d\times d$ density operator $\bsigma$.
  Set
  \[
    t_0 = \frac{3\ln 2}{2}\cdot \frac{1}{4\alpha_2(\LL_\bsigma)}
    \ln{\frac{q-1}{p-1}}.
  \]
  Then, for every $d^n\times d^n$ operator $\mat{A}$ acting on $n$ qudits and every
  $t\ge t_0$, the QMS $\Phi_t=\romanE^{-t\LL_\bsigma}$ generated by the simple Lindbladian $\LL_\bsigma$ satisfies
  \[
    \qnn{\Phi_{t}^{\otimes n}(\mat{A})}\leq \pnn{\mat{A}}.
  \]
\end{theorem}

We first observe that, without loss of generality, the following
conditions hold:
\begin{itemize}
  \item $\bsigma$ is a full-rank diagonal density operator.
  \item $\mat{A}$ is positive semi-definite.

\end{itemize}
The first condition is valid because we can always choose a basis that diagonalizes $\bsigma$. The second condition follows from a generalization of the result of Watrous~\cite{10.5555/2011608.2011614}
and Gupta and Wilde~\cite{gupta2015multiplicativity}. They proved similar
results in the unbiased case where $\bsigma = \id/d$. Here we generalize this result
to the biased case where $\bsigma$ can be any full-rank density operator. Also, the
following lemma is stated in greater generality for any completely positive linear
map, which includes the generalized depolarizing channel.
This follows directly from \cite[Theorem 1]{10.5555/2011608.2011614}
and that the map $\Gamma^{1/q}\circ\Xi\circ \Gamma^{-1/p}$ is completely positive if and only if $\Xi$ is completely positive.
For completeness, we provide an independent proof in Appendix \ref{app:A}.

\begin{lemma}
  \label{lem:hypercontractivity-only-for-PSD}
  For integers $d\ge2$ and $n\ge1$, a full-rank density operator $\bsigma$ on $\mathbb C^d$, and a completely positive linear map
  $\Xi:\Matrix_{d^n}\to\Matrix_{d^n}$, the following reduction holds for all $1\le p\le q$:
  \begin{align*}
    &\sup\set{\qnn{\Xi(\mat{X})}: \mat{X}\in\Matrix_{d^n}, \pnn{\mat{X}} = 1}
    \\ =&
    \sup\set{\qnn{\Xi(\mat{X})}: \mat{X}\in\PSD_{d^n}, \pnn{\mat{X}} = 1}.
  \end{align*}
\end{lemma}

\subsection{Norm Compression Inequality}
In this subsection we prove the norm compression inequality.

\begin{theorem}
  \label{thm:qudit-norm-compression}
  For an integer $p\ge2$ and integers $d,n\ge1$, consider a $dn\times dn$
  positive semi-definite matrix $\mat{M}$ written as a $d\times d$ block matrix
  \[
    \mat{M} =
    \begin{bmatrix}
      \mat{M}_{11} & \cdots & \mat{M}_{1d} \\
      \vdots & \ddots & \vdots \\
      \mat{M}_{d1} & \cdots & \mat{M}_{dd} \\
    \end{bmatrix}.
  \]
  Define the $d\times d$ matrix
  \[
    \mat{m} =
    \begin{bmatrix}
      \normsub{\mat{M}_{11}}{p} & \cdots & \normsub{\mat{M}_{1d}}{p} \\
      \vdots              & \ddots & \vdots              \\
      \normsub{\mat{M}_{d1}}{p} & \cdots & \normsub{\mat{M}_{dd}}{p} \\
    \end{bmatrix}.
  \]
  Then
  \begin{equation}
    \label{eq:qudit-norm-compression}
    \normsub{\mat{M}}{p}\le \normsub{\mat{m}}{p}.
  \end{equation}
\end{theorem}

\begin{remark}
  \label{remark:norm_compression_false_for_qudit} The $d=2$ case was proved by~\cite{king2003inequalities}.
  For the $d=2$ case, King has also proved the following:
  \begin{itemize}
    \item The matrix $\mat{m}$ is positive semi-definite.

    \item \cref{eq:qudit-norm-compression} holds for all real $p\ge 2$ and for all
      $1 \le p \le 2$ in the reverse direction.
  \end{itemize}
  However, all these facts do not hold when the local dimension $d>2$. That is,
  there are counterexamples showing that
  \begin{itemize}
    \item For $p=1.5, d=4$, there is a $4\times 4$ matrix such that
      \cref{eq:qudit-norm-compression} does not hold in the reverse direction~\cite{AUDENAERT2006155,AUDENAERT2008781}.

    \item For $p=2.5, d=4$, there is a $4\times 4$ matrix such that
      \cref{eq:qudit-norm-compression} does not hold. See \cref{example:counter2.5}.

    \item The following is not true for $p=4/3$: if $p = \frac{2t}{2t-1}$ for
      some integer $t\ge 1$, it holds that
      \begin{equation*}
        \label{eq:qudit-norm-compression-ge}\normsub{\mat{M}}{p}\ge \normsub{\mat{m}}{p}.
      \end{equation*}
      See \cref{example:counter4/3}.

    \item For $p=3, d=3$, there is a $9\times 9$ matrix $\mat{M}$ such that the
      corresponding matrix $\mat{m}$ is not positive semi-definite. Also for $p=4, d=3$,
      and a $6\times 6$ matrix.
      See \cref{example:counter99} and \cref{example:counter66}.
  \end{itemize}
\end{remark}

\begin{proof}[Proof of \cref{thm:qudit-norm-compression}]
  Note that
  \begin{align*}
    \normsub{\mat{M}}{p} & = \br{\Tr\Br{\mat{M}^p}}^{1/p}                                                                           \\
                   & = \br{\sum_{i_1, \dots, i_p}\Tr\Br{\mat{M}_{i_1i_2}\cdot \mat{M}_{i_2i_3}\cdot\dots\cdot \mat{M}_{i_pi_1}}}^{1/p}    \\
                   & \le \br{\sum_{i_1, \dots, i_p}\normsub{\mat{M}_{i_1i_2}}{p}\cdot\dots\cdot\normsub{\mat{M}_{i_pi_1}}{p}}^{1/p} \\
                   & = \br{\Tr\Br{\mat{m}^{p}}}^{1/p}
  \end{align*}
  The first equality follows because $\mat{M}$ is positive semi-definite. The
  inequality follows from \Holder's inequality. Now if $p$ is an even integer,
  \begin{equation*}
    \br{\Tr\Br{\mat{m}^{p}}}^{1/p}= \br{\Tr\Br{\br{\mat{m}\mat{m}^\dagger}^{p/2}}}^{1/p}= \normsub{\mat{m}}{p}.
  \end{equation*}
  For $p$ being odd, since $\mat{m}$ is a real symmetric matrix, it
  can be written as $\mat{m} = \mat{a} - \mat{b}$, where $\mat{a}$ and $\mat{b}$ are both positive semi-definite
  real matrices, $\mat{a}\mat{b} = \mat{b}\mat{a} = 0$ and $\abs{\mat{m}}= \mat{a} + \mat{b}$. We have
  \begin{align*}&
    \br{\Tr\Br{\mat{m}^{p}}}^{1/p}= \br{\Tr\Br{(\mat{a}^p-\mat{b}^p)}}^{1/p}\le \br{\Tr\Br{(\mat{a}^p+\mat{b}^p)}}
    ^{1/p}= \normsub{\mat{m}}{p}.\qedhere
  \end{align*}

\end{proof}

The above norm compression inequality can be adapted to the biased case.
\begin{cor}
  \label{cor:norm_compression_non_com}
  For an integer $p\ge2$, an integer $n\ge1$, and a full-rank diagonal density operator
  $\bsigma\in\Matrix_d$, every PSD $d^n\times d^n$ matrix $\mat{M}$ satisfies
  \[
    \pnn{\mat{M}}= \norm{\fmat{\mat{M}_{11}}{\mat{M}_{1d}}{\mat{M}_{d1}}{\mat{M}_{dd}}}_{\bsigma^{\otimes n},p}\leq
    \pn{{\fmat{\pnnn{\mat{M}_{11}}}{\pnnn{\mat{M}_{1d}}}{\pnnn{\mat{M}_{d1}}}{\pnnn{\mat{M}_{dd}}}}}.
  \]
\end{cor}

\begin{proof}
  Since $\bsigma$ is diagonal, by \cref{thm:qudit-norm-compression}, we have
  \begin{align*}
    \pnn{\mat{M}}= & \pnorm{\fmat{\sigma_{1}^{1/p}\gt{1/p}{\mat{M}_{11}}}{\br{\sigma_1\sigma_d}^{1/2p}\gt{1/p}{\mat{M}_{1d}}}{\br{\sigma_d\sigma_1}^{1/2p}\gt{1/p}{\mat{M}_{d1}}}{\sigma_d^{1/p}\gt{1/p}{\mat{M}_{dd}}}}                                 \\
    \le      & \pnorm{\fmat{\sigma_{1}^{1/p}\pnorm{\gt{1/p}{\mat{M}_{11}}}}{\br{\sigma_1\sigma_d}^{1/2p}\pnorm{\gt{1/p}{\mat{M}_{1d}}}}{\br{\sigma_d\sigma_1}^{1/2p}\pnorm{\gt{1/p}{\mat{M}_{d1}}}}{\sigma_d^{1/p}\pnorm{\gt{1/p}{\mat{M}_{dd}}}}} \\
    =        & \pnorm{\bsigma^{1/2p}\fmat{\pnnn{\mat{M}_{11}}}{\pnnn{\mat{M}_{1d}}}{\pnnn{\mat{M}_{d1}}}{\pnnn{\mat{M}_{dd}}}\bsigma^{1/2p}}                                                                                                         \\
    =        & \pn{\fmat{\pnnn{\mat{M}_{11}}}{\pnnn{\mat{M}_{1d}}}{\pnnn{\mat{M}_{d1}}}{\pnnn{\mat{M}_{dd}}}}.
  \end{align*}
  Here, the notation $\gt{1/p}\cdot$ above is interpreted as the tensor product
  $\Gamma_{\bsigma^{\otimes (n-1)}}^{1/p}=(\Gamma_{\bsigma}^{1/p})^{\otimes (n-1)}$.
\end{proof}

\subsection{Integer Hypercontractivity}
The first inequality we derive towards \cref{thm:qudit_hc} is a $\br{q, 2}$-hypercontractivity statement for integer $q$. This is achieved by an inductive argument. The $n=1$ case is given by Beigi, Datta, and Rouz{\'e}~\cite{beigiQuantumReverseHypercontractivity2020}.
The induction step follows from the norm compression inequality.

\begin{theorem}[Integer Tensorized Hypercontractivity]\label{thm:integer_hc_noncom}
    Let $q\ge 2$ and $n\ge1$ be integers, and let $\bsigma$ be a full-rank density operator
 on $\mathbb C^d$. Set
    \(
      t_0=\frac{1}{4\alpha_2(\LL_\bsigma)}\ln(q-1).
    \)
    Then, for every $d^n\times d^n$ operator $\mat{A}$ acting on $n$ qudits and every
    $\rho\le \romanE^{-t_0}$, the generalized depolarizing channel satisfies
    \[
      \qnn{\DepBChn(\mat{A})}\leq \tnn{\mat{A}}.
    \]
\end{theorem}

\begin{proof}
  Without loss of generality, we can assume that $\bsigma=\diag(\sigma_1,\ldots,\sigma_d)$ is a diagonal density operator.
  We proceed by induction on~$n$. The base case follows from the one-site value of $\alpha_2(\LL_\bsigma)$ together with
\cref{thm:logarithmic-Sobolev-and-hypercontractivity-equivalent}.
  For $n \ge 2$, assume that the inequality holds for $n-1$.
  We then write $\mat{A}$ as a $d \times d$ block matrix:
\begin{equation*}
    \mat{A} = \fmat{\mat{A}_{11}}{\mat{A}_{1d}}{\mat{A}_{d1}}{\mat{A}_{dd}}.
\end{equation*}
  Let $\mat{B}_{ij} = \DepBCh^{\otimes n-1}\br{\mat{A}_{ij}}$, and denote $\widetilde{\mat{B}} = ({\qnnn{\mat{B}_{ij}}})_{i,j}$.
  We have
  \begin{equation*}
    \qnn{\DepBChn(\mat{A})} =
    \qnn{\fmat{\mat{C}_{11}}{\mat{C}_{1d}}{\mat{C}_{d1}}{\mat{C}_{dd}}},
  \end{equation*}
  where
  $\mat{C}_{ii} = [\rho+(1-\rho)\sigma_i]\mat{B}_{ii} + \sum_{j\neq i}(1-\rho)\sigma_j\mat{B}_{jj}\succeq 0$
  and $\mat{C}_{ij} = \rho \mat{B}_{ij}$ for $i\neq j$.

  By the triangle inequality, we have
  \begin{equation*}
    \qnnn{\mat{C}_{ii}} \le [\rho+(1-\rho)\sigma_i]\qnnn{\mat{B}_{ii}} + \sum_{j\neq i}(1-\rho)\sigma_j\qnnn{\mat{B}_{jj}}=(\DepBCh(\widetilde{\mat{B}}))_{ii}.
  \end{equation*}
Using \cref{lem:diag-monotonicity-psd}, we can find appropriate $\epsilon_i\geq 0$ such that
\[
\qnnn{\mat{C}_{ii}+\epsilon_i\id}=(\DepBCh(\widetilde{\mat{B}}))_{ii}.
\]
Notice that
$$\qnnn{\mat{C}_{ij}}=\rho\qnnn{\mat{B}_{ij}}=(\DepBCh(\widetilde{\mat{B}}))_{ij}.$$
Using \cref{lem:diag-monotonicity-psd} again and \cref{cor:norm_compression_non_com}, we obtain
  \begin{equation}
  \begin{aligned}
      \qnn{\DepBChn(\mat{A})}
    &\leq
    \qnn{\fmat{\mat{C}_{11}+\epsilon_1\id}{\mat{C}_{1d}}{\mat{C}_{d1}}{\mat{C}_{dd}+\epsilon_d\id}}\\
    &\leq \qn{\fmat{\qnnn{\mat{C}_{11}+\epsilon_1\id}}{\qnnn{\mat{C}_{1d}}}{\qnnn{\mat{C}_{d1}}}{\qnnn{\mat{C}_{dd}+\epsilon_d\id}}}\\
    &= \qn{\DepBCh\fmat{\qnnn{\mat{B}_{11}}}{\qnnn{\mat{B}_{1d}}}{\qnnn{\mat{B}_{d1}}}{\qnnn{\mat{B}_{dd}}}}\\
    &= \qn{\DepBCh(\widetilde{\mat{B}})}.
  \end{aligned}
  \end{equation}

  By the base case, we have
  \begin{equation*}
    \qnn{\DepBChn(\mat{A})}\leq \qn {\DepBCh(\widetilde{\mat{B}}) } \leq \tn{\widetilde{\mat{B}}} = \tn{\fmat{\qnnn{\mat{B}_{11}}}{\qnnn{\mat{B}_{1d}}}{\qnnn{\mat{B}_{d1}}}{\qnnn{\mat{B}_{dd}}}}.
  \end{equation*}

  Recall that $\mat{B}_{ij} = \DepBCh^{\otimes n-1}\br{\mat{A}_{ij}}$.
  So we apply the induction hypothesis for $n-1$ to each $\qnnn{\mat{B}_{ij}}$,
  which yields $\qnnn{\mat{B}_{ij}} \le \tnnn{\mat{A}_{ij}}$.
  Combined with \cref{lem:monotone_biased_norm} for $p=2$ and \cref{lem:exact_norm_compression_biased_2}, we have
  \begin{align*}
   &  \qnn{\DepBChn(\mat{A})}\leq \tn{\fmat{\tnnn{\mat{A}_{11}}}{\tnnn{\mat{A}_{1d}}}{\tnnn{\mat{A}_{d1}}}{\tnnn{\mat{A}_{dd}}}}=\tnn{\mat{A}}. \qedhere
  \end{align*}
\end{proof}

The inductive step of \cref{thm:integer_hc_noncom} can be generalized into the
following tensorization result.

\begin{theorem}\label{cor:two_tensor_hc}
  Fix an integer $q\ge2$. Let $\Psi:\mathcal B(\mathcal H)\to\mathcal B(\mathcal H)$
  be a \PODS\ channel with respect to the eigenbasis of a full-rank state $\bsigma$.
  Assume that
  \(
    \norm{\Psi(\mat{A})}_{\bsigma,q}\leq \norm{\mat{A}}_{\bsigma,2},
     \mat{A}\in\mathcal B(\mathcal H).
  \)
  Let $\Psi':\mathcal B(\mathcal H')\to\mathcal B(\mathcal H')$ be a channel such that,
  for a full-rank state $\bomega$ on $\mathcal H'$,
  \(
    \norm{\Psi'(\mat{B})}_{\bomega,q}\leq \norm{\mat{B}}_{\bomega,2},
     \mat{B}\in\mathcal B(\mathcal H').
  \)
  Then, for every $\mat{X}\in\mathcal B(\mathcal H\otimes\mathcal H')$,
  \[
    \norm{(\Psi\otimes\Psi')(\mat{X})}_{\bsigma\otimes\bomega,q}
    \leq
    \norm{\mat{X}}_{\bsigma\otimes\bomega,2}.
  \]
\end{theorem}

\begin{proof}
By \cref{lem:hypercontractivity-only-for-PSD}, it suffices to consider $\mat{X}\succeq0$.
Write $\mat{X}=(\mat{X}_{ij})_{i,j=1}^d$ as a block matrix over the eigenbasis of $\bsigma$,
and set $\mat{B}_{ij}=\Psi'(\mat{X}_{ij})$ and
\(
  \widetilde{\mat{B}}=\bigl(\norm{\mat{B}_{ij}}_{\bomega,q}\bigr)_{i,j=1}^d .
\)
Using the \PODS\ form of $\Psi$, write
\[
  \Psi(\ketbratwo{i}{j})=\gamma_{i,j}\ketbratwo{i}{j}\quad(i\neq j),
  \qquad
  \Psi(\ketbra{i})=\sum_{j=1}^d\eta_{i,j}\ketbra{j}.
\]
The blocks of $(\Psi\otimes\Psi')(\mat{X})$ are then
$\mat{C}_{ii}=\sum_j\eta_{j,i}\mat{B}_{jj}$ and $\mat{C}_{ij}=\gamma_{i,j}\mat{B}_{ij}$ for $i\neq j$.
Since the \PODS\ coefficients are nonnegative, the triangle inequality gives
\begin{align*}
  \norm{\mat{C}_{ii}}_{\bomega,q}
  &\leq \sum_j\eta_{j,i}\norm{\mat{B}_{jj}}_{\bomega,q}
    =(\Psi(\widetilde{\mat{B}}))_{ii},\\
  \norm{\mat{C}_{ij}}_{\bomega,q}
  &=\gamma_{i,j}\norm{\mat{B}_{ij}}_{\bomega,q}
    =(\Psi(\widetilde{\mat{B}}))_{ij}.
\end{align*}
As in the proof of \cref{thm:integer_hc_noncom}, we may increase the diagonal
blocks if necessary. Applying \cref{lem:diag-monotonicity-psd} and
\cref{cor:norm_compression_non_com} then yields
\[
  \norm{(\Psi\otimes\Psi')(\mat{X})}_{\bsigma\otimes\bomega,q}
  \leq
  \norm{\Psi(\widetilde{\mat{B}})}_{\bsigma,q}
  \leq
  \norm{\widetilde{\mat{B}}}_{\bsigma,2}.
\]
By the assumption on $\Psi'$, entrywise
$\norm{\mat{B}_{ij}}_{\bomega,q}\le\norm{\mat{X}_{ij}}_{\bomega,2}$. Hence
\cref{lem:monotone_biased_norm} and \cref{lem:exact_norm_compression_biased_2}
give
\[
  \norm{\widetilde{\mat{B}}}_{\bsigma,2}
  \leq
  \norm{\bigl(\norm{\mat{X}_{ij}}_{\bomega,2}\bigr)_{i,j}}_{\bsigma,2}
  =
  \norm{\mat{X}}_{\bsigma\otimes\bomega,2}.
\]
\end{proof}

It follows from \cref{cor:two_tensor_hc} that the integer tensorized
hypercontractivity result remains valid for all \PODS\ channels.
\begin{cor}\label{thm:integer_hc_noncom_general}
  Suppose that $\Psi:\mathcal B(\mathbb C^d)\to\mathcal B(\mathbb C^d)$ is a
  \PODS\ channel with respect to the eigenbasis of a full-rank state $\bsigma$.
  Fix an integer $q\ge2$. If
  \[
    \qn{\Psi(\mat{A})}\leq \tn{\mat{A}},\qquad \mat{A}\in\mathcal B(\mathbb C^d),
  \]
  then, for every $n\ge1$ and every
  $\mat{A}\in\mathcal B((\mathbb C^d)^{\otimes n})$,
  \[
    \qnn{\Psi^{\otimes n}(\mat{A})}\leq \tnn{\mat{A}}.
  \]
\end{cor}

\subsection{Complex Interpolation}

Using complex interpolation, we extend the integer
hypercontractivity from the previous subsection to real exponents. Similar ideas are used in
\cite{diaconisLogarithmicSobolevInequalities1996} to derive classical logarithmic-Sobolev
inequalities, and also appear in the work of Olkiewicz and Zegarlinski~\cite{olkiewiczHypercontractivityNoncommutativeLpSpaces1999},
and in the work of Temme, Pastawski, and Kastoryano~\cite{Temme2014Hypercontractivity}
in the matrix setting.

\begin{theorem}
  \label{thm:complex-interpolation}
  Consider a primitive $\bsigma$-reversible Lindblad generator $\LL$, where $\bsigma$ is a
  full-rank $d\times d$ density operator, and write $\Psi_t=\romanE^{-t\LL}$ for the associated
  QMS. Given $2<r<\infty$, if there exists $t_r>0$ such that
  \[
    \norm{\Psi_{t_r}(\mat{A})}_{\bsigma,r}\le \tn{\mat{A}}
  \]
  for every operator $\mat{A}$, then
  \[
    \alpha_2(\LL)\ge \frac{1-2/r}{2t_r}.
  \]
\end{theorem}

\begin{remark}

It is proved in~\cite[Theorem 5]{Temme2014Hypercontractivity} that if $\normsub{\Psi_{t_r}}{\bsigma, 2\to r}
\le M_{r}$ for some $M_{r}>0$, the logarithmic-Sobolev constant is lower bounded
by
\[
  \alpha_{2}(\LL)\ge \frac{\lambda(1-2/r)}{2\br{\lambda t_r+\log M_r+\br{r-2}/r}}
  ,
\]
where $M_{r}$ can be arbitrarily large. When $M_{r}= 1$, their
logarithmic-Sobolev lower bound is
\[
  \frac{1-2/r}{2\br{t_r+\br{r-2}/(r\lambda)}},
\]
which is weaker than our bound in \cref{thm:complex-interpolation}. This
is because our \cref{thm:complex-interpolation} only deals with the case
$M_{r}= 1$, which leads directly to $\normsub{\Psi_{t_r}}{\bsigma, 2\to r}\le 1$
(see \cite[Eq. (18) and (19)]{Temme2014Hypercontractivity}), without invoking Rothaus'
lemma (\cite[Theorem 4.2]{olkiewiczHypercontractivityNoncommutativeLpSpaces1999}).
This saves the $\log(M_{r})$ term and removes the term
$(r-2)/(r\lambda)$, which makes our estimate independent of the spectral gap.

\end{remark}

Our proof uses the fact that the $\bsigma$-weighted $L_{p}$ space
satisfies the
following complex interpolation relation, which we state here for completeness.

\begin{lemma}[Complex Interpolation, see e.g. {\cite[Theorem 2]{beigiSandwichedRenyiDivergence2013}}]
  \label{lem:com_interpolation_noncom}
  Fix a full-rank density operator $\bsigma$. For a continuous map
  $F:\{0\le\Re z\le1\}\to B(\mathcal H)$ that is holomorphic on the interior
  $\{0<\Re z<1\}$ and continuous on the boundary, and for parameters $\theta\in[0,1]$ and $p\le q\le r$ satisfying
  \[
    \frac{1}{q}=\frac{1-\theta}{p}+\frac{\theta}{r},
  \]
  one has
  \[
    \norm{F(\theta)}_{\bsigma,q}
    \le
    \sup_{t\in\mathbb R}\norm{F(it)}_{\bsigma,p}^{1-\theta}
    \cdot
    \sup_{t\in\mathbb R}\norm{F(1+it)}_{\bsigma,r}^{\theta}.
  \]
  In particular, for $F(z)=\romanE^{-zs\LL}(\mat{A})$ with $s\in\mathbb R$,
  \[
    \norm{\romanE^{-\theta s\LL}(\mat{A})}_{\bsigma,q}
    \le
    \sup_{t\in\mathbb R}\norm{\romanE^{-(1+it)s\LL}(\mat{A})}_{\bsigma,r}^{\theta}
    \cdot
    \sup_{t\in\mathbb R}\norm{\romanE^{-\iu t s\LL}(\mat{A})}_{\bsigma,p}^{1-\theta}.
  \]
\end{lemma}
We now give the proof of \cref{thm:complex-interpolation}.

\begin{proof}[Proof of \cref{thm:complex-interpolation}]
  Apply \cref{lem:com_interpolation_noncom}.
For any $q\in[2,r]$, let $\theta\in[0,1]$ satisfy
\begin{equation*}
  \frac{1}{q}= \frac{1-\theta}{2}+\frac{\theta}{r}.
\end{equation*}
By \cref{lem:com_interpolation_noncom} with $s=t_{r}$, we have
\begin{align*}
  \norm{\romanE^{-\theta t_r\LL}(\mat{A})}_{\bsigma,q}
  &\le
  \sup_{s\in\mathbb R}\norm{\romanE^{-(1+\iu s)t_r\LL}(\mat{A})}_{\bsigma,r}^{\theta}
  \cdot
  \sup_{s\in\mathbb R}\tn{\romanE^{-\iu s t_r\LL}(\mat{A})}^{1-\theta} \\
  &=
  \sup_{s\in\mathbb R}\norm{\romanE^{-t_r\LL}\circ \romanE^{-\iu s t_r\LL}(\mat{A})}_{\bsigma,r}^{\theta}
  \cdot
  \sup_{s\in\mathbb R}\tn{\romanE^{-\iu s t_r\LL}(\mat{A})}^{1-\theta} \\
  &\le
  \sup_{s\in\mathbb R}\tn{\romanE^{-\iu s t_r\LL}(\mat{A})}^{\theta}
  \cdot
  \sup_{s\in\mathbb R}\tn{\romanE^{-\iu s t_r\LL}(\mat{A})}^{1-\theta} \\
  &=
  \sup_{s\in\mathbb R}\tn{\romanE^{-\iu s t_r\LL}(\mat{A})}
  =
  \tn{\mat{A}}.
\end{align*}

  The first inequality is \cref{lem:com_interpolation_noncom}. The second
  inequality follows from the assumption. The last equality follows from \cref{lem:complex-depolarize-2norm-invariant_noncom}.

  Set $t=\theta t_r$. Then
\[
  q=q(t)=\frac{2rt_r}{rt_r-(r-2)t}.
\]
At $t=0$, we have
\[
  \norm{\romanE^{-t\LL}(\mat{A})}_{\bsigma,q(t)}=\tn{\mat{A}}.
\]
Applying
\cref{thm:hc_to_lsi} with input exponent $p=2$ and hypercontractive time
\[
  t=t(q)=\frac{rt_r(q-2)}{q(r-2)}
\]
gives
\[
  \alpha_2(\LL)
  \ge
  \frac{1}{4}\left.\frac{dq}{dt}\right|_{t=0}
  =
  \frac{1-2/r}{2t_r}.
\]

\end{proof}

\subsection{General \texorpdfstring{($q,p$)}{(q,p)}-Hypercontractivity}
We now give the proof of \cref{thm:qudit_hc}.

\begin{proof}[ Proof of \cref{thm:qudit_hc}]

  By \cref{prop:QMS}, the map $\Phi_{t}^{\otimes n}$ is a contraction with
  respect to $\qn{\cdot}$ for all $q\ge 1$, so we only need to prove the desired inequality
  in the $t=t_{0}$ case. Then by \cref{thm:integer_hc_noncom}, the hypercontractivity
  holds for integer values. Specifically, fix $r = 3$ and let $t_{r}= \frac{1}{4\alpha_{2}(\LL_\bsigma)}
  \ln(r-1) = \frac{\ln 2}{4\alpha_{2}(\LL_\bsigma)}$. Then by
  \cref{thm:complex-interpolation},
  \begin{equation}
    \alpha_{2}(\LL_\bsigma^{(n)}) \ge \frac{1-2/r}{2t_{r}}= \frac{2}{3\ln 2}\alpha_{2}(\LL_\bsigma
    )\approx 0.96\alpha_{2}(\LL_\bsigma).\label{eq:lsi_tensorisation}
  \end{equation}
  Finally, we conclude the proof by \cref{thm:logarithmic-Sobolev-and-hypercontractivity-equivalent} and \cref{lem:strook-varopoulos_non_com}.
\end{proof}

For a general primitive $\bsigma$-reversible QMS $(\Psi_t)_{t\ge0}$ satisfying
\PODS, \cref{lem:pods-upgrades-reversibility} shows that the QMS is strongly
$\bsigma$-reversible. Consequently, \cref{cor:two_tensor_hc} and
\cref{thm:complex-interpolation} imply the following tensorization property for
the generator $\LL$, comparable to
\cite[Theorem 24]{beigiQuantumReverseHypercontractivity2020}:
\begin{theorem}\label{lem:two_tensor_lsi}
Let $\bsigma,\bomega$ be full-rank density operators. Suppose that
$(\Psi_t)_{t\ge0}$ is generated by a primitive $\bsigma$-reversible Lindblad
generator $\LL$, while $(\Psi'_t)_{t\ge0}$ is generated by a primitive strongly
$\bomega$-reversible Lindblad generator $\LL'$.
If $(\Psi_t)_{t\ge0}$ satisfies \PODS\ with respect to the eigenbasis of $\bsigma$, then
\[
\alpha_2(\LL\otimes \id' + \id\otimes \LL')
\geq
\frac{2}{3\ln 2}\min\bigl\{\alpha_2(\LL),\alpha_2(\LL')\bigr\}.
\]
Moreover, the assumptions on $\LL$ and $(\Psi_t)_{t\ge0}$ alone imply that,
for every integer $n\ge1$,
\begin{equation}
\alpha_2(\LL^{(n)})
\ge \frac{2}{3\ln 2}\,\alpha_2(\LL)
\approx 0.96\,\alpha_2(\LL).
\label{eq:lsi_tensorisation_general}
\end{equation}
\end{theorem}

\begin{proof}
Set
\[
  a=\min\bigl\{\alpha_2(\LL),\alpha_2(\LL')\bigr\},
  \qquad
  t_3=\frac{\ln 2}{4a}.
\]
By \cref{lem:pods-upgrades-reversibility}, $\LL$ is strongly
$\bsigma$-reversible. Hence \cref{lem:strook-varopoulos_non_com} and
\cref{thm:logarithmic-Sobolev-and-hypercontractivity-equivalent}, first for
$\LL$ and then for $\LL'$, give
\[
  \norm{\Psi_{t_3}(\mat{A})}_{\bsigma,3}
  \le \norm{\mat{A}}_{\bsigma,2},
  \qquad
  \norm{\Psi'_{t_3}(\mat{B})}_{\bomega,3}
  \le \norm{\mat{B}}_{\bomega,2}.
\]
The inequalities initially obtained for positive semidefinite inputs extend to
all operators by \cref{lem:hypercontractivity-only-for-PSD}. Since
$\Psi_{t_3}$ is \PODS, \cref{cor:two_tensor_hc} yields
\[
  \norm{(\Psi_{t_3}\otimes\Psi'_{t_3})(\mat{X})}_{\bsigma\otimes\bomega,3}
  \le \norm{\mat{X}}_{\bsigma\otimes\bomega,2}.
\]
The product generator $\LL\otimes\id'+\id\otimes\LL'$ is primitive and
$(\bsigma\otimes\bomega)$-reversible. Applying
\cref{thm:complex-interpolation} with $r=3$ proves
\[
  \alpha_2(\LL\otimes\id'+\id\otimes\LL')
  \ge \frac{1-2/3}{2t_3}
  =\frac{2}{3\ln2}\,a.
\]

For the $n$-fold conclusion, instead set
$t_3=\ln2/(4\alpha_2(\LL))$. The same one-site argument and
\cref{thm:integer_hc_noncom_general} give the $(3,2)$-hypercontractive
inequality for $\Psi_{t_3}^{\otimes n}$ directly. A single application of
\cref{thm:complex-interpolation} to the primitive reversible generator
$\LL^{(n)}$ gives \eqref{eq:lsi_tensorisation_general}.
\end{proof}

As a corollary, we obtain the following hypercontractivity statement:
\begin{cor}
  \label{thm:qudit_hc_general}
  Fix $1\leq p\leq q$, an integer $n\ge1$, and a full-rank $d\times d$ density operator $\bsigma$.
  Suppose that $\LL$ is a primitive $\bsigma$-reversible Lindblad generator, and write
  $\Psi_t=\romanE^{-t\LL}$ for the generated QMS. If $(\Psi_t)_{t\ge0}$ satisfies \PODS\ with respect to the eigenbasis of $\bsigma$, then,
  with
  \[
    t_{0}= \frac{3\ln 2}{2}\cdot \frac{1}{4\alpha_{2}(\LL)}
    \ln{\frac{q-1}{p-1}},
  \]
  where $\alpha_{2}(\LL)$ is the logarithmic-Sobolev constant of $\LL$, one has
  \[
    \qnn{\Psi_t^{\otimes n}(\mat{A})}\leq \pnn{\mat{A}}
  \]
  for every $d^{n}\times d^{n}$ operator $\mat{A}$ acting on $n$ qudits and every $t\ge t_0$.
\end{cor}

\begin{proof}
By \cref{lem:pods-upgrades-reversibility}, $\LL$ is strongly
$\bsigma$-reversible; hence $\LL^{(n)}$ is primitive and strongly
$\bsigma^{\otimes n}$-reversible. By \cref{lem:two_tensor_lsi},
\[
  \alpha_2(\LL^{(n)})
  \ge \frac{2}{3\ln2}\alpha_2(\LL).
\]
Now \cref{lem:strook-varopoulos_non_com} and
\cref{thm:logarithmic-Sobolev-and-hypercontractivity-equivalent} give the
claimed bound for positive semidefinite inputs at every $t\ge t_0$, and
\cref{lem:hypercontractivity-only-for-PSD} extends it to all operators.
\end{proof}

The qubit case admits an exact refinement of
\cref{lem:two_tensor_lsi,thm:qudit_hc_general}. Indeed, King's norm compression
inequality~\cite[Theorem~1]{king2003inequalities} holds for every real $q\ge2$,
so the proof of \cref{cor:two_tensor_hc} applies without the integer
restriction. Combining this with
\cref{lem:pods-upgrades-reversibility,lem:strook-varopoulos_non_com} and the
standard equivalence between LSI and hypercontractivity yields the following
real-exponent hypercontractive refinement of
\cref{thm:qubit-pods-exact-tensorization}.

\begin{theorem}[Real-exponent hypercontractivity for qubit \PODS\ QMS]
\label{thm:qubit-pods-real-hypercontractivity}
Let $\bsigma$ be a full-rank qubit state and let $\LL$ be a primitive
$\bsigma$-reversible Lindblad generator on
$\mathcal B(\mathbb C^2)$, and write $\Psi_t=\romanE^{-t\LL}$. Suppose that
$\Psi_t$ is \PODS\ with respect to an eigenbasis of $\bsigma$ for every
$t\ge0$. Then, for every integer $n\ge1$ and every real $q\ge2$, the
following tensorized hypercontractive estimate holds:
\[
  \norm{\Psi_{t(q)}^{\otimes n}(\mat{A})}_{\bsigma^{\otimes n},q}
  \le \norm{\mat{A}}_{\bsigma^{\otimes n},2},
  \qquad
  t(q)=\frac{\ln(q-1)}{4\alpha_2(\LL)},
\]
for every $\mat{A}\in\mathcal B((\mathbb C^2)^{\otimes n})$.
Consequently,
\[
  \alpha_2(\LL^{(n)})=\alpha_2(\LL),
\]
recovering \cref{thm:qubit-pods-exact-tensorization}.
\end{theorem}

For comparison, the works~\cite{beigiQuantumReverseHypercontractivity2020,Temme2014Hypercontractivity}
take a different route to tensorization for general reversible QMS, based on the
tensorization of the spectral gap. Suppose that $\LL$ is a primitive
$\bsigma$-reversible Lindblad generator, so that the Schr\"odinger-picture semigroup admits
$\bsigma$ as its unique invariant density state. The \emph{spectral gap} of $\LL$ is
  \[
    \lambda(\LL)=\inf_{\mat{X},\svar{\mat{X}}>0}{\frac{\diri{2}{\mat{X}}}{\svar{\mat{X}}}},
  \]
  where $\svar{\mat{X}}=\spro{\mat{X}}{\mat{X}}-\spro{\mat{X}}{\id}^{2}=\tn{\mat{X}}^{2}-\spro{\mat{X}}{\id}^{2}$. The
  spectral gap always satisfies the tensorization property:
  \[
    \lambda(\nll)=\lambda(\LL).
  \]
  It is standard (e.g. \cite{Temme2014Hypercontractivity})  that we have
  $\lambda(\LL)\geq 2\alpha_{2}(\LL)$. Following this approach in \cite{beigiQuantumReverseHypercontractivity2020}, we obtain the following tensorization property:

\begin{cor}
  \label{cor:tensorized_lsi_general}
  Every primitive $\bsigma$-reversible Lindblad generator $\LL$ satisfies
  \begin{equation}\label{eq:spectral}
\alpha_2(\LL^{(n)})
\ge
\frac{2}{3\ln2}\alpha_2(\LL_\bsigma)\lambda(\LL)
\ge
\frac{4}{3\ln2}\alpha_2(\LL_\bsigma)\alpha_2(\LL),
  \end{equation}
  where $\LL_\bsigma$ is the simple Lindblad generator associated with $\bsigma$ and $\lambda(\LL)$ is the spectral gap of $\LL$.
\end{cor}

 \begin{proof}
Let
\[
    E_\bsigma(\mat{X})=\Tr(\bsigma \mat{X})\id,
    \qquad
    \LL_\bsigma(\mat{X})=\mat{X}-E_\bsigma(\mat{X}).
\]
For each $1\le i\le n$, write
\[
    \LL_i=\id^{\otimes(i-1)}\otimes \LL\otimes \id^{\otimes(n-i)},
    \qquad
    \LL_{\bsigma,i}=\id^{\otimes(i-1)}\otimes \LL_\bsigma\otimes \id^{\otimes(n-i)},
\]
and
\[
    E_{\bsigma,i}
    =
    \id^{\otimes(i-1)}\otimes E_\bsigma\otimes \id^{\otimes(n-i)}.
\]
Applying the logarithmic-Sobolev estimate \eqref{eq:lsi_tensorisation} to
$\LL_\bsigma^{(n)}=\sum_{i=1}^{n}\LL_{\bsigma,i}$ gives, for all
$\mat{X}\in B(\mathcal{H}^{\otimes n})$,
\[
    C\ent{2}{\mat{X}}
    \leq
    \sum_{i=1}^{n}\mathcal{E}_{\LL_{\bsigma,i}}(\mat{X})
    =
    \sum_{i=1}^{n}\spro{\mat{X}}{\LL_{\bsigma,i}(\mat{X})},
\]
where
\(
    C=\frac{2}{3\ln 2}\alpha_2(\LL_\bsigma).
\)

    Since the spectral gap tensorizes, $\lambda(\LL_{i})=\lambda(\LL)$ for each $i$, and we have
    \[
      \spro{\mat{X}}{\LL_i(\mat{X})}\ge \lambda(\LL) \norm{ \mat{X}-E_{\bsigma, i}(\mat{X})}_{\bsigma^{\otimes n},
      2}^{2}=\lambda(\LL) \spro{\mat{X}}{\LL_{\bsigma,i}(\mat{X})}.
    \]
    Summing over $1\le i\le n$, we have
    \[
      \mathcal{E}_{\LL^{(n)}}(\mat{X})=\sum_{i=1}^{n}\spro{\mat{X}}{\LL_i(\mat{X})}\ge \lambda(\LL
      ) \sum_{i=1}^{n}\spro{\mat{X}}{\LL_{\bsigma,i}(\mat{X})}\ge C\lambda(\LL)\ent{2}{\mat{X}},
    \]
    which proves the claim.
  \end{proof}

  In a similar form, \cite{Temme2014Hypercontractivity} gives another spectral-gap-based
  quasi tensorization property:

\begin{theorem}[{\cite[Theorem 9]{Temme2014Hypercontractivity}}]
  \label{thm:tpk14_tensor}
  Every primitive $\bsigma$-reversible Lindblad generator $\LL$ satisfies
  \begin{align*}
    \alpha_{2}(\LL^{(n)})\geq
    \frac{1}{2}\cdot
    \frac{\lambda(\LL)}{\ln(\lambda(\bsigma)^{-1})+11+4\ln d}
    \geq
    \frac{\alpha_{2}(\LL)}{\ln(\lambda(\bsigma)^{-1})+11+4\ln d}.
  \end{align*}
\end{theorem}

The estimates in \cref{cor:tensorized_lsi_general} and
\cref{thm:tpk14_tensor} apply to general primitive reversible QMS, whereas
\eqref{eq:lsi_tensorisation_general} applies to the more structured primitive \PODS\
class. Thus these bounds should be compared with their scopes in mind. In the
\PODS\ setting, \eqref{eq:lsi_tensorisation_general} gives the direct dimension-free estimate
\[
    \alpha_2(\LL^{(n)})\ge \frac{2}{3\ln2}\alpha_2(\LL),
\]
without reducing through the spectral gap. For general primitive reversible QMS,
\cref{cor:tensorized_lsi_general} gives a better spectral-gap-based tensorization
estimate than \cref{thm:tpk14_tensor}. The rigorous proof of this fact is given in Appendix \ref{app:comparison}.

\section{Optimal \texorpdfstring{($q,2$)}{(q,2)}-Hypercontractivity for Generalized Qubit Depolarizing Channels}\label{sec:q_2_qubit_hc}

In this section, we present an optimal ($q,2$)-hypercontractivity result for the generalized qubit depolarizing channel that matches classical results. It suffices to restrict attention to the case where $\bsigma$ is diagonal and $\mat{A}$ is positive semidefinite. Note that in this section, we do not require $q$ to be an integer.

For a $2\times 2$ PSD matrix $\mat{A}=\ffmat abdc$, and $\bsigma = \ffmat{\mu}{}{}{1-\mu}$, the generalized qubit depolarizing channel can be expressed as
    \begin{align*}
    \DepBCh (\mat{A}) &=\rho \mat{A}+(1-\rho)\Tr(\mat{A}\bsigma)\id \\
    &=\ffmat{\rho a+(1-\rho)(\mu a+(1-\mu)c)}{\rho b}{\rho d}{\rho c+(1-\rho)(\mu a+(1-\mu)c)},
    \end{align*}
    where $0\leq \mu\leq \frac 1 2$ and $d=\overline b$.

We first recall the corresponding classical two-point estimate for diagonal self-adjoint inputs.

\begin{theorem}[Classical \texorpdfstring{($q,2$)}{(q,2)} HC, \cite{pawelHypercontractivitySimpleRandom2007,oleszkiewicz2003nonsymmetric}]\label{thm:classical_q_2_hc}
    If $\mat{A}=\ffmat a00c$ is diagonal and self-adjoint, then the following hypercontractivity holds for the generalized qubit depolarizing channel $\DepBCh$, where $\mu\leq 1/2$ is the smaller eigenvalue of full-rank density operator $\bsigma$:
    \[
        \qn{\DepBCh(\mat{A})}\leq \tn{\mat{A}},
    \]
    for
    $$0\leq \rho\leq \rho_\mu=\sqrt{\frac{(1-\mu)^{\frac 2 q}-\mu^{\frac 2 q}}{\mu^{\frac 2q-1}(1-\mu)-(1-\mu)^{\frac 2 q-1}\mu}}.$$
    Moreover, $\rho_\mu$ is optimal for $q\geq 2$.
\end{theorem}

We next pass from the classical diagonal estimate to the genuinely quantum qubit case. The following theorem shows that, for positive semidefinite $2\times2$ inputs, the off-diagonal entries do not reduce the optimal admissible value of $\rho$ from the classical two-point bound.

\begin{theorem}[Single-Site Generalized Qubit Depolarizing \texorpdfstring{($q,2$)}{(q,2)} HC]\label{thm:single_qubit_q_2_hc}
    For a $2\times 2$ full-rank density operator $\bsigma$, denote $\mu\leq 1/2$ is the smaller eigenvalue of $\bsigma$. Given $0\leq \rho\leq \rho_\mu=\sqrt{\frac{(1-\mu)^{\frac 2 q}-\mu^{\frac 2 q}}{\mu^{\frac 2q-1}(1-\mu)-(1-\mu)^{\frac 2 q-1}\mu}}$, the following hypercontractivity holds for the generalized qubit depolarizing channel $\DepBCh$ and every $2\times 2$ PSD matrix $\mat{A}$; moreover, $\rho_\mu$ is optimal for $q\geq 2$:
    \[
        \qn{\DepBCh (\mat{A})}\leq \tn{\mat{A}}.
    \]
\end{theorem}

Before proving \cref{thm:single_qubit_q_2_hc}, we first introduce two technical lemmas that will be used in the proof, whose proofs are deferred to Appendix \ref{app:easy} and \ref{app:easy_2}.
\begin{lemma}
  \label{lemma:easy} For $x > y\geq 0$, $x+y>0$ and $q\geq 2$, we have
  \[
    \frac{1}{q-1}\cdot \frac{x^{q-1}-y^{q-1}}{x-y}\leq \br{\frac{x^{q}+y^{q}}{2}}^{\frac{q-2}{q}}
    .
  \]
\end{lemma}

\begin{lemma}
  \label{lem:easy_2} For $\mu\in (0,\frac{1}{2}]$, $q\geq 2$, we have
  \[
    \sqrt{\frac{(1-\mu)^{\frac{2}{q}}-\mu^{\frac{2}{q}}}{\mu^{\frac{2}{q}-1}(1-\mu)-(1-\mu)^{\frac{2}{q}-1}\mu}}
    =\rho_{\mu}\leq \rho'=\sqrt{\frac{(4\mu(1-\mu))^{\frac{1}{2}-\frac{1}{q}}}{q-1}}
    .
  \]
\end{lemma}

Now we give the proof of \cref{thm:single_qubit_q_2_hc}:

\begin{proof}[Proof of \cref{thm:single_qubit_q_2_hc}]
Denote
\[
  \ffmat{\tilde a}{\tilde b}{\overline{\tilde b}}{\tilde c}
  =
  \bsigma^{\frac{1}{2q}}\DepBCh(\mat{A})\bsigma^{\frac{1}{2q}},
\]
  we have
  \[
    \begin{aligned}
      \tilde a &= \mu^{1/q}((\rho+(1-\rho)\mu)a+(1-\rho)(1-\mu)c),\\
      \tilde b &= (\mu(1-\mu))^{\frac{1}{2q}}\rho b, \\
      \tilde c &= (1-\mu)^{1/q}((1-\rho)\mu a+(\rho +(1-\rho)(1-\mu))c).
    \end{aligned}
  \]

  Define
  \[
    f(u)=\left(\left(\frac{\tilde a+\tilde c+\sqrt{u}}{2}\right)^{q}+\left(\frac{\tilde
    a+\tilde c-\sqrt{u}}{2}\right)^{q}\right)^{\frac{2}{q}},
  \]
  where
  \[
    u=(\tilde a-\tilde c)^{2}+4(\mu(1-\mu))^{1/q}\rho^{2}|b|^{2}.
  \]

  All we need is $f(u)\leq \mu a^{2}+(1-\mu)c^{2}+2\sqrt{\mu(1-\mu)}|b|^{2}$. From the classical
  case, i.e. \cref{thm:classical_q_2_hc}, this is true for $|b|^{2}=0$. Therefore, we
  only need to show that the derivative of $f(u)$ with respect to $|b|^{2}$ is bounded above by $2\sqrt{\mu(1-\mu)}$ for $|b|^2 > 0$.
Without loss of generality, we assume that $b\in\rr$, and we next prove
  \begin{equation}
\frac{\partial}{\partial b^{2}}f(u)\leq 2\sqrt{\mu(1-\mu)}.\label{eq:base_quantum}
  \end{equation}

  The left-hand side equals
  \begin{dmath*}
    \frac{\partial u}{\partial b^{2}}\cdot \frac{\partial}{\partial u}f(u)=(4[\mu(1-\mu)]^{\frac{1}{q}}\rho^2)\cdot
    \frac{2}{q}\left(\left(\frac{\tilde a+\tilde c+\sqrt{u}}{2}\right)^q+\left(\frac{\tilde
    a+\tilde c-\sqrt{u}}{2}\right)^q\right)^{\frac{2}{q}-1} \cdot
    \frac{q}{4\sqrt{u}}\left(\left(\frac{\tilde a+\tilde c+\sqrt{u}}{2}\right)^{q-1}-\left(\frac{\tilde
    a+\tilde c-\sqrt{u}}{2}\right)^{q-1}\right).
  \end{dmath*}

  Denote
  \[
    x=\frac{\tilde a+\tilde c+\sqrt{u}}{2}, \quad y=\frac{\tilde a+\tilde c-\sqrt{u}}{2}
    ,
  \]
  then $x-y=\sqrt{u}\geq 0$. All we need is
  \begin{gather*}
   2[\mu(1-\mu)]^{\frac{1}{q}}\rho^{2}(x^{q}+y^{q})^{\frac{2-q}{q}}\cdot \frac{x^{q-1}-y^{q-1}}{x-y}
    \leq 2\sqrt{\mu(1-\mu)}\\
    \iff \tilde \rho \cdot \frac{x^{q-1}-y^{q-1}}{x-y}\leq (x^{q}+y^{q})^{\frac{q-2}{q}}
    ,
  \end{gather*}
  where $\tilde \rho=(\mu(1-\mu))^{\frac{1}{q}-\frac{1}{2}}\rho^{2}$. Using Lemma
  \ref{lemma:easy}, it remains to prove
  \(
    \tilde \rho\leq \frac{2^{\frac{q-2}{q}}}{q-1},
  \)
  which is equivalent to
  \[
  \rho\leq \rho'=\sqrt{\frac{[4\mu(1-\mu)]^{\frac{1}{2}-\frac{1}{q}}}{q-1}}.
  \]
  By assumption, $\rho\leq \rho_\mu$. By \cref{lem:easy_2}, $\rho_\mu\leq \rho'$, and hence $\rho\leq \rho'$, which gives the desired derivative bound. The optimality of $\rho_\mu$ follows from the classical diagonal case~\cite{oleszkiewicz2003nonsymmetric}.
\end{proof}

Using norm compression for qubits~\cite{king2003inequalities} ($q\geq 2$ case) and \cref{cor:norm_compression_non_com}, we obtain hypercontractivity for the tensorized generalized qubit depolarizing channel. The proof is the same as \cref{thm:integer_hc_noncom}:

\begin{theorem}[Tensorized Generalized Qubit Depolarizing \texorpdfstring{($q,2$)}{(q,2)} HC]\label{thm:tensorized_qubit_q_2_hc}
    For a $2\times 2$ full-rank density operator $\bsigma$, denote $\mu\leq 1/2$ is the smaller eigenvalue of $\bsigma$. Given $0\leq\rho\leq \rho_\mu=\sqrt{\frac{(1-\mu)^{\frac{2}{q}}-\mu^{\frac{2}{q}}}{\mu^{\frac{2}{q}-1}(1-\mu)-(1-\mu)^{\frac{2}{q}-1}\mu}}$, the following hypercontractivity holds for every operator $\mat{A}$ acting on $n$ qubits; moreover, $\rho_\mu$ is optimal for $q\geq 2$:
    \[
        \qnn{\DepBChn(\mat{A})}\leq \tnn{\mat{A}}.
    \]
\end{theorem}

\section{Application: Kahn--Kalai--Linial Theorem for Qudit Depolarizing Channels}

We now spell out how the hypercontractivity theorem for product depolarizing
channels proved in the previous sections feeds into quantum influence theory.
We use the abstract framework of~\cite{rouze2024quantumkkl}. To make the
connection explicit, we recall the assumptions used there. Let
\((P_t)_{t\ge0}=\romanE^{-t\LL}\) be a KMS-symmetric QMS on a von Neumann algebra
\((\mathcal M,\varphi)\), let \(\mathcal A\) be a core algebra, and let
\(i_2:\mathcal M\to L_2(\mathcal M,\varphi)\) denote the symmetric embedding.
The hypotheses are as follows:
\begin{itemize}
    \item[\textbf{(H0)}] there exists a \(*\)-subalgebra
    \(\mathcal A\subset D(L)\) which is weakly dense in \(\mathcal M\) and
    invariant under \((P_t)_{t\ge0}\);

    \item[\textbf{(H1)}] the carr\'e du champ operator satisfies a Bakry--Emery type
    gradient estimate
    \[
        \Gamma(P_t x)\le \romanE^{-2Kt}P_t(\Gamma(x)),
        \qquad x\in\mathcal A,\ t\ge0;
    \]

    \item[\textbf{(H2)}] there are self-adjoint coordinate maps
    \(d_j:\mathcal A\to\mathcal M\), \(j\in J\), such that
    \[
        \langle i_2(x),i_2(Lx)\rangle
        =
        \sum_{j\in J}\|i_2(d_jx)\|_2^2,
    \]
    and, for some constant \(M>0\),
    \[
        \max_{j\in J}\|d_jx\|\le M\|\Gamma(x)\|^{1/2};
    \]

    \item[\textbf{(H3)}] a Poincar\'e inequality holds:
    \[
        \lambda\|i_2(x-\varphi(x)\mathbf 1)\|_2^2
        \le
        \langle i_2(x),i_2(Lx)\rangle;
    \]

    \item[\textbf{(H4)}] a hypercontractive estimate holds: for some
    \(\alpha_{\mathrm{H4}}>0\),
    \[
        \|i_2(P_t x)\|_2 \le \|i_{p(t)}(x)\|_{p(t)},\qquad
        p(t)=1+\romanE^{-2\alpha_{\mathrm{H4}}t};
    \]

    \item[\textbf{(H5)}] the semigroup and coordinate maps satisfy an
    intertwining estimate: for some \(\mu\in\mathbb R\),
    \[
        \|i_p(d_jP_t x)\|_p
        \le
        \romanE^{-\mu t}\|i_p(P_t d_jx)\|_p,
        \qquad p\in[1,\infty];
    \]

    \item[\textbf{(H6)}] for every \(I\subseteq J\), if \(\mat{E}_I\) denotes the
    orthogonal projection onto \(\bigcap_{i\in I}\ker d_i\) in
    \(L_2(\mathcal M,\varphi)\), then a restricted Poincar\'e inequality holds:
    \[
        \nu\|i_2(x)-\mat{E}_I(i_2(x))\|_2^2
        \le
        \sum_{i\in I}\|i_2(d_i x)\|_2^2.
    \]
\end{itemize}

For the product qudit depolarizing semigroup
\[
    P_t=\Phi_t^{\otimes n}=\romanE^{-t\LL_\bsigma^{(n)}},
\]
these hypotheses take a concrete form. We work on the finite-dimensional algebra
\(\mathcal M=M_d^{\otimes n}\) with
\(\varphi(x)=\Tr(\bsigma^{\otimes n}x)\), so \textbf{(H0)} is automatic with
\(\mathcal A=\mathcal M\). The coordinate maps in \textbf{(H2)} are
\[
    d_j=\LL_{\bsigma,j}
    =
    \id^{\otimes(j-1)}\otimes\LL_\bsigma\otimes\id^{\otimes(n-j)}.
\]
The gradient estimate \textbf{(H1)} comes from \cite{junge2015noncommutative} for $K=1/2$. The Poincar\'e inequality \textbf{(H3)} could be derived from \textbf{(H4)} (see e.g. \cite{olkiewiczHypercontractivityNoncommutativeLpSpaces1999}).
Since the coordinate maps commute with the product semigroup,
\(
    d_jP_t=P_td_j,
\)
the intertwining estimate \textbf{(H5)} holds with $\mu=0$. The restricted
Poincar\'e condition \textbf{(H6)} is
also automatic: if
\[
    \mathbb E_j
    =
    \id^{\otimes(j-1)}
    \otimes (\mat{X}\mapsto \Tr(\bsigma \mat{X})\mathbf 1)
    \otimes \id^{\otimes(n-j)},
    \qquad
    \LL_{\bsigma,j}=\id-\mathbb E_j,
\]
then the maps \(\mathbb E_j\) are commuting orthogonal projections on
\(L_2(\bsigma^{\otimes n})\), and hence
\(\mat{E}_I=\prod_{j\in I}\mathbb E_j\) is the projection onto
\(\bigcap_{j\in I}\ker \LL_{\bsigma,j}\); this gives \textbf{(H6)} with
\(\nu=1\).
Thus the only nontrivial analytic input in our setting is \textbf{(H4)}.
Our contribution is precisely to verify this hypercontractive input for
\(P_t=\Phi_t^{\otimes n}\), uniformly in \(n\) and in arbitrary local dimension.

The three consequences below are therefore direct applications of our
hypercontractivity theorem to the abstract machinery of~\cite{rouze2024quantumkkl}.

\begin{cor}[KKL inequality for qudit depolarizing channels]\label{cor:kkl_qudit}
Let $(P_t)_{t\ge 0}=\Phi_t^{\otimes n}$ be the generalized product qudit depolarizing semigroup. Then there exists a constant $C>0$ independent of $n$ such that for every self-adjoint $\mat{x}$ satisfying
\[
    \tnn{\mat{x}} = 1, \qquad \|\mat{x}\|_{\infty} \le 1, \qquad \Tr(\mat{x} \bsigma^{\otimes n}) = 0,
\]
one has
\[
    \max_{j\in [n]} \mathrm{Inf}_j^1(\mat{x}) \ge C \frac{\sqrt{\log n}}{n}.
\]
Here, $\mathrm{Inf}_j^p(\mat{x})=\pnn{\LL_{\bsigma,j}(\mat{x})}^p$ denotes the $p$-influence.
\end{cor}

\begin{proof}
As explained above, the product qudit depolarizing semigroup satisfies the
structural hypotheses required in~\cite{rouze2024quantumkkl}. The only
nontrivial analytic input is \textbf{(H4)}. By \cref{thm:qudit_hc}, \textbf{(H4)} holds with
\[
    \alpha_{\mathrm{H4}}
    \ge
    \frac{4}{3\ln2}\alpha_2(\LL_\bsigma).
\]
The KKL lower bound then follows from~\cite[Theorem 4.5]{rouze2024quantumkkl}.
\end{proof}

Our qudit hypercontractivity combined with \cite[Theorem 4.3]{rouze2024quantumkkl} yields the following Talagrand inequality for $1$-influence.

\begin{cor}[Talagrand inequality for qudit depolarizing channels]\label{cor:talagrand_qudit}
Let $(P_t)_{t\ge 0}=\Phi_t^{\otimes n}$ be the product qudit depolarizing semigroup with invariant state $\bsigma$. Then there exists a constant $C>0$ such that for every self-adjoint operator $\mat{x}$ with $\|\mat{x}\|\le 1$,
\[
    \|\mat{x}-\operatorname{Tr}(\mat{x}\bsigma^{\otimes n})\mathbf 1\|_{\bsigma^{\otimes n},2}^2
    \le C \sum_{j \in [n]}
    \frac{\mathrm{Inf}_j^1(\mat{x})(1+\mathrm{Inf}_j^1(\mat{x}))}
    {(1+\log^+(1/\mathrm{Inf}_j^1(\mat{x})))^{1/2}}.
\]
Here $\mathrm{Inf}_j^p(\mat{x})=\pnn{\LL_{\bsigma,j}(\mat{x})}^p$ denotes the $p$-influence.
\end{cor}

\begin{proof}
The structural assumptions are verified as above, and \cref{thm:qudit_hc}
provides the hypercontractive hypothesis \textbf{(H4)}. The result follows from
the abstract Talagrand inequality~\cite[Theorem 4.3]{rouze2024quantumkkl}.
\end{proof}

Finally, once the hypercontractive estimate has been established, the structural
assumptions recalled above lead to a junta approximation theorem. We specialize
\cite[Theorem 4.9]{rouze2024quantumkkl} to $\nu=1$ and $\mu=0$. In the
convention below, \(\mat{E}_I\) averages out the coordinates in \(I\), so the effective
junta is supported on the complement \([n]\setminus I\).

\begin{cor}[Junta theorem for qudit depolarizing channels]\label{cor:junta_qudit}
Let $(P_t)_{t\ge 0}=\Phi_t^{\otimes n}$ be the product qudit depolarizing semigroup. Then there exists a constant $C>0$ such that, for every self-adjoint operator $\mat{x}$ with $\|\mat{x}\|\le 1$ and every $0<\varepsilon\le 2$, one can find a set $I\subseteq [n]$ satisfying
\[
    \|\mat{x}-\mat{E}_I(\mat{x})\|_{2,\bsigma^{\otimes n}}\le \varepsilon
\]
and
\[
    |[n]\setminus I|
    \le
    \begin{cases}
        \mathrm{Inf}^{1}(\mat{x})^2\exp\!\left(C\dfrac{\mathrm{Inf}^{2}(\mat{x})}{\varepsilon^2}\log\!\left(\dfrac{2\mathrm{Inf}^{2}(\mat{x})}{\varepsilon}\right)\right), & \mathrm{Inf}^{2}(\mat{x})\ge 1,\\[1.2ex]
        \dfrac{\mathrm{Inf}^{1}(\mat{x})^2}{\mathrm{Inf}^{2}(\mat{x})}\exp\!\left(C\dfrac{\mathrm{Inf}^{2}(\mat{x})}{\varepsilon^2}\log\!\left(\dfrac{2\sqrt{\mathrm{Inf}^{2}(\mat{x})}}{\varepsilon}\right)\right), & \mathrm{Inf}^{2}(\mat{x})<1.
    \end{cases}
\]
Here \(\mat{E}_I\) is the \(L_2(\bsigma^{\otimes n})\)-orthogonal projection onto
\(\bigcap_{i\in I}\ker \LL_{\bsigma,i}\), equivalently the conditional expectation
that averages out the coordinates in \(I\). Moreover,
\[
    \mathrm{Inf}^{p}(\mat{x}):=\sum_{j\in [n]}\mathrm{Inf}_j^p(\mat{x}),
    \qquad
    \mathrm{Inf}_j^p(\mat{x})=\|\LL_{\bsigma,j}(\mat{x})\|_{\bsigma^{\otimes n},p}^p
\]
are the total and coordinate \(p\)-influences, respectively.
\end{cor}

\begin{proof}
The junta theorem in~\cite[Theorem 4.9]{rouze2024quantumkkl} applies to a
KMS-symmetric QMS satisfying \textbf{(H0)}, \textbf{(H2)}, and
\textbf{(H4)}--\textbf{(H6)}. For the product qudit depolarizing semigroup, the
structural hypotheses were verified above, and \cref{thm:qudit_hc} provides the
remaining hypercontractive input \textbf{(H4)}. Therefore,
\cite[Theorem 4.9]{rouze2024quantumkkl} yields the displayed junta bound.
\end{proof}

\section*{Acknowledgements}
YD, FO, PY, and HZ were supported by the National Natural Science Foundation of China (Grant Nos. 62332009 and 12347104), the Quantum Science and Technology---National Science and Technology Major Project (Grant No. 2021ZD0302901), the NSFC/RGC Joint Research Scheme (Grant No. 12461160276), the Natural Science Foundation of Jiangsu Province (No. BK20243060), the Fundamental and Interdisciplinary Disciplines Breakthrough Plan of the Ministry of Education of China (No. JYB2025XDXM118), the ``111 Center'' (No. B26023), and the Fundamental Research Funds for the Central Universities (Grant No. 2026300376). LG is partially supported by the National Natural Science Foundation of China (grant No.~12401163) and the Department of Science and Technology of Hubei Province (Project No. 2025EHA041, Project No. 2025AFA044).

\bibliographystyle{amsplain}
\bibliography{references}

\appendix
\section{\texorpdfstring{Proof of \cref{lem:hypercontractivity-only-for-PSD}}{Proof
of Lemma hypercontractivity-only-for-PSD}}
\label{app:A}
\begin{proof}
  Since $\Xi$ is completely positive, it is defined by a set of operators
  $\set{\mat{A}_k}_{k=1}^{M}$ such that
  \[
    \Xi(\mat{X}) = \sum_{k=1}^{M}\mat{A}_{k}\mat{X}\mat{A}_{k}^{\dagger}.
  \]
  Let $\mat{X}\in\Matrix_{d^n}$ be any operator satisfying $\pn{\mat{X}}= 1$. Let $\mat{Y}\in\Matrix
  _{d^n}$ be the normalized tangent functional of $\Xi(\mat{X})$ with respect to $\qn{\cdot}$.
  That is, let $q^{*}$ be the \Holder\ conjugate of $q$ which satisfies $\frac{1}{q^{*}}
  + \frac{1}{q}= 1$. We choose $\mat{Y}$ such that
  \begin{equation*}
    \normsub{\mat{Y}}{\bsigma, q^*}=1 \quad\text{and}\quad\qn{\Xi(\mat{X})}= \abs{\spro{\mat{Y}}{\Xi(\mat{X})}}
    .
  \end{equation*}
  Apply singular value decomposition to $\gtt{p}{\mat{X}}$ and $\gtt{q^*}{\mat{Y}}$ as
  \begin{equation*}
    \gtt{p}{\mat{X}}= \sum_{i}s_{i}\ketbratwo{u_i}{v_i}\quad\text{and}\quad \gtt{q^*}{\mat{Y}}
    = \sum_{i}t_{i}\ketbratwo{w_i}{x_i}.
  \end{equation*}
  We have that
  \begin{equation*}
    \mat{X} = \sum_{i}s_{i}\bsigma^{-\frac{1}{2p}}\ketbratwo{u_i}{v_i}\bsigma^{-\frac{1}{2p}}
    \quad\text{and}\quad \mat{Y} = \sum_{i}t_{i}\bsigma^{-\frac{1}{2q^{*}}}\ketbratwo{w_i}
    {x_i}\bsigma^{-\frac{1}{2q^{*}}}.
  \end{equation*}
  Define positive semi-definite operators $\mat{X}_{L}, \mat{X}_{R}, \mat{Y}_{L}, \mat{Y}_{R}$ as
  \begin{equation*}
    \mat{X}_{L} = \sum_{i}s_{i}\bsigma^{-\frac{1}{2p}}\ketbra{u_i}\bsigma^{-\frac{1}{2p}}
    , \quad \mat{Y}_{L} = \sum_{i}t_{i}\bsigma^{-\frac{1}{2q^{*}}}\ketbra{w_i}\bsigma^{-\frac{1}{2q^{*}}}
    ,
  \end{equation*}
  \begin{equation*}
    \mat{X}_{R} = \sum_{i}s_{i}\bsigma^{-\frac{1}{2p}}\ketbra{v_i}\bsigma^{-\frac{1}{2p}}
    , \quad \mat{Y}_{R} = \sum_{i}t_{i}\bsigma^{-\frac{1}{2q^{*}}}\ketbra{x_i}\bsigma^{-\frac{1}{2q^{*}}}
    .
  \end{equation*}
  It is immediate from the definition that
  \begin{equation*}
    \pn{\mat{X}_L}= \pn{\mat{X}_R}= 1 \quad\text{and}\quad \normsub{\mat{Y}_L}{\bsigma, q^*}= \normsub
    {\mat{Y}_R}{\bsigma, q^*}= 1.
  \end{equation*}
  Now
  \begin{align*}
    \qn{\Xi(\mat{X})} & =\abs{\spro{\mat{Y}}{\Xi(\mat{X})}}                                                                                                                                                                                                                                 \\
                & = \abs{\sum_{k}\Tr\Br{\bsigma^{\frac{1}{2}}\mat{Y}^\dagger\bsigma^{\frac{1}{2}}\mat{A}_k\mat{X}\mat{A}_k^\dagger}}                                                                                                                                                                \\
                & = \abs{\sum_{ijk}s_it_j\Tr\Br{\bsigma^{\frac{1}{2}-\frac{1}{2q^{*}}}\ketbratwo{x_j}{w_j}\bsigma^{\frac{1}{2}-\frac{1}{2q^{*}}}\mat{A}_k\bsigma^{-\frac{1}{2p}}\ketbratwo{u_i}{v_i}\bsigma^{-\frac{1}{2p}}\mat{A}_k^\dagger}}                                            \\
                & = \abs{\sum_{ijk}s_it_j\bra{w_j}\bsigma^{\frac{1}{2}-\frac{1}{2q^{*}}}\mat{A}_k\bsigma^{-\frac{1}{2p}}\ket{u_i}\cdot\bra{v_i}\bsigma^{-\frac{1}{2p}}\mat{A}_k^\dagger\bsigma^{\frac{1}{2}-\frac{1}{2q^{*}}}\ket{x_j}}                                                   \\
                & \le \sqrt{\sum_{ijk}s_{i}t_{j}\abs{\bra{w_j}\bsigma^{\frac{1}{2}-\frac{1}{2q^*}}\mat{A}_k\bsigma^{-\frac{1}{2p}}\ket{u_i}}^{2}}\cdot\sqrt{\sum_{ijk}s_{i}t_{j}\abs{\bra{v_i}\bsigma^{-\frac{1}{2p}}\mat{A}_k^\dagger\bsigma^{\frac{1}{2}-\frac{1}{2q^*}}\ket{x_j}}^{2}} \\
                & = \sqrt{\spro{\mat{Y}_L}{\Xi(\mat{X}_L)}}\cdot \sqrt{\spro{\mat{Y}_R}{\Xi(\mat{X}_R)}}                                                                                                                                                                                          \\
                & \le \sqrt{\qn{\Xi(\mat{X}_L)}}\cdot \sqrt{\qn{\Xi(\mat{X}_R)}}.
  \end{align*}
  Here, the first inequality is Cauchy--Schwarz, and the second follows from
  \Holder's inequality. The above inequality implies that either $\qn{\Xi(\mat{X})}\le \qn{\Xi(\mat{X}_L)}$
  or $\qn{\Xi(\mat{X})}\le \qn{\Xi(\mat{X}_R)}$. Notice that $\mat{X}_{L}$ and $\mat{X}_{R}$ are both
  positive semi-definite, so this concludes the proof.
\end{proof}

\section{Comparison of Spectral-Gap-Based Tensorization Constants}\label{app:comparison}

\begin{proof}
Writing
$\mu=\lambda(\bsigma)$ and
\[
    a_\bsigma:=\alpha_2(\LL_\bsigma)
    =
    \frac{1-2\mu}{\ln(\mu^{-1}-1)},
\]
the ratio between the coefficient in \cref{cor:tensorized_lsi_general} and that
in \cref{thm:tpk14_tensor} is
\[
    R
    =
    \frac{\frac{2}{3\ln2}a_\bsigma}
    {\frac{1}{2}\frac{1}{\ln(\mu^{-1})+11+4\ln d}}
    =
    \frac{4}{3\ln2}
    \frac{1-2\mu}{\ln(\mu^{-1}-1)}
    \bigl(\ln(\mu^{-1})+11+4\ln d\bigr).
\]
Set $y=\mu^{-1}-1$. Since $0<\mu\le 1/d\le 1/2$, we have
$y\ge d-1\ge1$, and
\[
    R
    =
    \frac{4}{3\ln2}
    \frac{y-1}{(y+1)\ln y}
    \bigl(\ln(y+1)+11+4\ln d\bigr),
\]
with the value at $y=1$ understood by continuity. Since $\ln(y+1)>\ln y$ for
$y>1$, it is enough to lower bound
\[
    g(y)=\frac{y-1}{y+1}\left(1+\frac{C}{\ln y}\right),
    \qquad C=11+4\ln d.
\]
For $1<y\le25$, using $\ln y\le y-1$ gives
\[
    g(y)>
    \frac{C}{y+1}
    \ge
    \frac{11+4\ln2}{26}
    >
    \frac{3\ln2}{4}.
\]
For $y>25$,
\[
    g(y)>
    \frac{y-1}{y+1}>
    \frac{24}{26}>
    \frac{3\ln2}{4}.
\]
The limiting case $y=1$ also satisfies the same bound by continuity. Hence
$R>1$, so the spectral-gap coefficient in
\cref{cor:tensorized_lsi_general} improves the coefficient in
\cref{thm:tpk14_tensor}.
\end{proof}

\section{Proof of \texorpdfstring{\cref{lemma:easy}}{Lemma 4.3}}\label{app:easy}

\begin{proof}
  \cref{lemma:easy} is equivalent to
  \[
    2b\cdot \left(\frac{(a+b)^{q}+(a-b)^{q}}{2}\right)^{\frac{q-2}{q}}-\frac{1}{q-1}
    \cdot((a+b)^{q-1}-(a-b)^{q-1})\geq 0,
  \]
  where $a=\frac{x+y}{2}>0,b=\frac{x-y}{2} > 0$.
  At $b=0$, the desired inequality holds. Hence it suffices to prove that the derivative with respect to $b$ is nonnegative for $b>0$, namely
  \begin{dmath*}
    \frac{q-2}{q}\cdot 2b\cdot \left(\frac{(a+b)^{q}+(a-b)^{q}}{2}\right)^{-\frac{2}{q}}\cdot
    \frac{q\left((a+b)^{q-1}-(a-b)^{q-1}\right)}{2}+2\left(\frac{(a+b)^{q}+(a-b)^{q}}{2}\right)^{\frac{q-2}{q}}\geq
    (a+b)^{q-2}+(a-b)^{q-2}.
  \end{dmath*}
  Note
  \[
    \frac{q-2}{q}\cdot 2b\cdot \left(\frac{(a+b)^{q}+(a-b)^{q}}{2}\right)^{-\frac{2}{q}}
    \cdot \frac{q\left((a+b)^{q-1}-(a-b)^{q-1}\right)}{2}\geq 0,
  \]
  and from the Generalized Mean Inequality, we have
  \[
    \left(\frac{(a+b)^{q}+(a-b)^{q}}{2}\right)^{\frac{q-2}{q}}\geq \frac{(a+b)^{q-2}+(a-b)^{q-2}}{2}
    ,
  \]
    which proves the lemma.
\end{proof}

\section{Proof of \texorpdfstring{\cref{lem:easy_2}}{Lemma 4.4}}\label{app:easy_2}

\begin{proof}
  Define $r = \frac{2}{q}$. Since $q \ge 2$, it follows that $r \in (0, 1]$.
  Let $x = 1 - \mu$ and $y = \mu$. Given $\mu \le 1/2$, we have $x \ge y > 0$ and the
  normalization condition $x + y = 1$.
  Now,
  \[
    \rho_{1}^{2} = (xy)^{1-r}\frac{x^{r} - y^{r}}{x^{2-r}- y^{2-r}}, \quad
    \rho_{2}^{2} = \frac{r}{2-r}(4xy)^{\frac{1-r}{2}}.
  \]

  Recall the definition of the Stolarsky mean for two distinct positive numbers $x
  , y$ and parameter $\alpha \neq 0, 1$ (see
  \cite{stolarsky1975generalizations,qi2002logarithmic}):
  \[
    S_{\alpha}(x, y) = \left( \frac{x^{\alpha} - y^{\alpha}}{\alpha(x-y)}\right)^{\frac{1}{\alpha-1}}
  \]
  and
  \[
    S_{0}(x,y)=\frac{x-y}{\ln x-\ln y}
  \]
  be the logarithmic mean.
  From this definition, we can express the difference of powers as:
  \[
    x^{\alpha} - y^{\alpha} = \alpha(x-y) S_{\alpha}(x,y)^{\alpha-1}.
  \]
  Applying this identity to both the numerator (with $\alpha = r$) and the denominator
  (with $\alpha = 2-r$) of $\rho_{1}^{2}$,
  we have
  \[
    \rho_{1}^{2} = (xy)^{1-r}\frac{r(x-y) S_{r}^{r-1}}{(2-r)(x-y) S_{2-r}^{1-r}}=
    \frac{r}{2-r}(xy)^{1-r}(S_{r} S_{2-r})^{r-1}.
  \]

  We wish to prove that $\rho_{1}^{2} \le \rho_{2}^{2}$, which now takes the
  form:
  \[
    \frac{r}{2-r}(xy)^{1-r}(S_{r} S_{2-r})^{r-1}\le \frac{r}{2-r}(4xy)^{\frac{1-r}{2}}
    \iff G^{2(1-r)}(S_{r} S_{2-r})^{r-1}\le (2G)^{1-r}.
  \] where $G = \sqrt{xy}$.

  For $r = 1$ (i.e., $q = 2$), the exponent $1-r = 0$, and the inequality
  trivially holds as an equality. For $r \in (0, 1)$, $1-r > 0$. We can raise both
  sides to the power of $\frac{1}{1-r}$ and rearranging gives:
  \[
    S_{r} S_{2-r}\ge \frac{G}{2}.
  \]
  Because $x+y=1$, we know $A = (x+y)/2 = 1/2$. Thus, we can write $\frac{G}{2}$ exactly as
  $A \cdot G$. The required inequality is therefore equivalent to:
  \[
    S_{r}(x, y) S_{2-r}(x, y) \ge A(x, y) G(x, y).
  \]

By the log-concavity properties of the Stolarsky mean $S_{\alpha}$ on
$\alpha\in[0,2]$ (see \cite{qi2002logarithmic}), for $\alpha_1+\alpha_2=2$ we have
  \[
    S_{r} S_{2-r}\geq S_{0}S_{2} = LA \ge GA,
  \]
  where $S_{0}=L$ is the logarithmic mean and $S_{2}=A$. Thus we finish the
  proof.
\end{proof}

\section{Norm Compression Counter-Examples}

\begin{example}\label{example:counter2.5}
  Let $p=2.5$. Let the matrix $\mat{M}$ be
  \begin{equation*}
      \mat{M} = \begin{bmatrix}
        0.22 & 0 & -0.18 & 0.14 \\
        0 & 0.21 & -0.17 & -0.14 \\
        -0.18 & -0.17 & 0.29 & 0 \\
        0.14 & -0.14 & 0 & 0.27
      \end{bmatrix}.
  \end{equation*}
  Then we have
  \begin{equation*}
        \mat{m} = \begin{bmatrix}
        0.22 & 0 & 0.18 & 0.14 \\
        0 & 0.21 & 0.17 & 0.14 \\
        0.18 & 0.17 & 0.29 & 0 \\
        0.14 & 0.14 & 0 & 0.27
      \end{bmatrix}.
  \end{equation*}
  It can be checked that $\mat{M}$ is positive semi-definite.
  Moreover,
  $\normsub{\mat{M}}{2.5} = 0.625943 \ge \normsub{\mat{m}}{2.5} = 0.622851$.
\end{example}

\begin{example}\label{example:counter4/3}
  Let $p=4/3$. Let the matrix $\mat{M}$ be
  \begin{equation*}
      \mat{M} = \begin{bmatrix}
        0.2 & -0.16 & 0.15 & 0 \\
        -0.16 & 0.29 & 0.01 & -0.19 \\
        0.15 & 0.01 & 0.27 & -0.18 \\
        0 & -0.19 & -0.18 & 0.24
      \end{bmatrix}.
  \end{equation*}
  Then we have
  \begin{equation*}
      \mat{m} = \begin{bmatrix}
        0.2 & 0.16 & 0.15 & 0 \\
        0.16 & 0.29 & 0.01 & 0.19 \\
        0.15 & 0.01 & 0.27 & 0.18 \\
        0 & 0.19 & 0.18 & 0.24
      \end{bmatrix}.
  \end{equation*}
  It can be checked that $\mat{M}$ is positive semi-definite.
  Moreover,
  $\normsub{\mat{M}}{4/3} = 0.83171 \le \normsub{\mat{m}}{4/3} = 0.884911$.
\end{example}

\begin{example}\label{example:counter99}
  Let $p=3$. Let the matrix $\mat{M}$ be
  \begin{equation*}
      \mat{M} = \begin{bmatrix}
        0.05 & -0.07 & 0 & 0.03 & 0.04 & -0.02 & 0.07 & -0.06 & 0.01 \\
        -0.07 & 0.19 & -0.01 & -0.02 & -0.01 & -0.02 & -0.07 & 0.19 & 0.02 \\
        0 & -0.01 & 0.02 & 0 & 0 & 0.01 & 0 & -0.02 & -0.01 \\
        0.03 & -0.02 & 0 & 0.06 & 0.07 & -0.04 & 0.09 & 0 & 0.01 \\
        0.04 & -0.01 & 0 & 0.07 & 0.19 & -0.09 & 0.16 & 0.06 & 0.07 \\
        -0.02 & -0.02 & 0.01 & -0.04 & -0.09 & 0.07 & -0.08 & -0.06 & -0.04 \\
        0.07 & -0.07 & 0 & 0.09 & 0.16 & -0.08 & 0.21 & -0.02 & 0.04 \\
        -0.06 & 0.19 & -0.02 & 0 & 0.06 & -0.06 & -0.02 & 0.25 & 0.07 \\
        0.01 & 0.02 & -0.01 & 0.01 & 0.07 & -0.04 & 0.04 & 0.07 & 0.07
      \end{bmatrix}.
  \end{equation*}
  Treating it as a $3\times 3$ block matrix, we have
  \begin{equation*}
      \mat{m} = \begin{bmatrix}
        0.219551 & 0.0578625 & 0.220537 \\
        0.0578625 & 0.269377 & 0.227928 \\
        0.220537 & 0.227928 & 0.315795
      \end{bmatrix}.
  \end{equation*}
  It can be checked that $\mat{M}$ is positive semi-definite,
  and that $\mat{m}$ is not positive semi-definite.
\end{example}
\begin{example}\label{example:counter66}
  Let $p=4$. Let the matrix $\mat{M}$ be
  \begin{equation*}
      \mat{M} = \begin{bmatrix}
        0.01 & 0.03 & 0 & 0 & -0.02 & 0 \\
        0.03 & 0.42 & 0 & 0 & -0.3 & -0.07 \\
        0 & 0 & 0.06 & -0.08 & 0.02 & -0.09 \\
        0 & 0 & -0.08 & 0.14 & -0.04 & 0.14 \\
        -0.02 & -0.3 & 0.02 & -0.04 & 0.25 & 0 \\
        0 & -0.07 & -0.09 & 0.14 & 0 & 0.18
      \end{bmatrix}.
  \end{equation*}
  Treating it as a $3\times 3$ block matrix, we have
  \begin{equation*}
      \mat{m} = \begin{bmatrix}
        0.422184 & 0 & 0.308674 \\
        0 & 0.189443 & 0.172274 \\
        0.308674 & 0.172274 & 0.265328
      \end{bmatrix}.
  \end{equation*}
  It can be checked that $\mat{M}$ is positive semi-definite,
  and that $\mat{m}$ is not positive semi-definite.
\end{example}

\end{document}